\documentclass[10pt,twocolumn]{article}

\usepackage[top=2.0cm,bottom=2.0cm,left=1.7cm,right=1.7cm,columnsep=0.6cm]{geometry}
\usepackage{amsmath,amssymb,amsthm}
\usepackage{physics}
\usepackage{mathtools}
\usepackage{bm}
\usepackage{booktabs,multirow,dcolumn,array}
\usepackage{graphicx}
\usepackage[table]{xcolor}
\usepackage{float,placeins}
\usepackage[toc,page]{appendix}
\usepackage{mdframed}
\usepackage{caption,subcaption}
\usepackage[T1]{fontenc}
\usepackage{enumitem}
\setlist{nosep,leftmargin=*}
\usepackage{hyperref}
\usepackage{cite}
\usepackage{cleveref}
\hypersetup{colorlinks=true,
  linkcolor={blue!60!black},citecolor={blue!60!black},urlcolor={blue!60!black},
  pdftitle={OAM-Induced Lattice Rotation Reveals a Fractional Optimum in Fault-Tolerant GKP Quantum Sensing},
  pdfauthor={Simanshu Kumar and Nandan S. Bisht}}
\usepackage{amsthm}
\newtheorem{proposition}{Proposition}

\newmdenv[backgroundcolor=gray!8,linecolor=black!40,linewidth=0.8pt,
  leftline=true,rightline=false,topline=false,bottomline=false,
  innerleftmargin=8pt,innerrightmargin=6pt,
  innertopmargin=5pt,innerbottommargin=5pt]{notebox}

\makeatletter
\renewcommand\section{\@startsection{section}{1}{\z@}%
  {-14pt \@plus -4pt}{5pt}{\normalfont\large\bfseries}}
\renewcommand\subsection{\@startsection{subsection}{2}{\z@}%
  {-10pt \@plus -3pt}{3pt}{\normalfont\normalsize\bfseries}}
\renewcommand\subsubsection{\@startsection{subsubsection}{3}{\z@}%
  {-8pt \@plus -2pt}{1pt}{\normalfont\normalsize\itshape}}
\makeatother

\newenvironment{findingbox}[1][]{%
  \smallskip\noindent\textbf{#1.}\enspace\ignorespaces}{\smallskip}

\newcommand{\QFI}{\mathcal{F}_{Q}}
\newcommand{\CFIM}{\mathcal{F}_{C}}
\newcommand{\Perr}{P_{\mathrm{err}}}
\newcommand{\Pthresh}{P_{\mathrm{th}}}
\newcommand{\SQL}{\mathrm{SQL}}
\newcommand{\HL}{\mathrm{HL}}
\newcommand{\GKP}{\mathrm{GKP}}
\newcommand{\SF}{\textsc{Strawberry Fields}}
\newcommand{\TF}{\textsc{TensorFlow}}
\newcommand{\etal}{\textit{et al.}}

\newcommand{\RR}{\mathbb{R}}
\newcommand{\CC}{\mathbb{C}}
\newcommand{\lvec}[1]{\mathbf{#1}}

\definecolor{RuleGray}{RGB}{180,185,190}
\graphicspath{{results/figures/}}
\begin{document}
\twocolumn[{%
\begin{center}
\vspace*{-2pt}
{\small\textbf{PREPRINT}\ $\cdot$\ Quantum Optics \& Quantum Information\ $\cdot$\ arXiv:2605.13271 [quant-ph]}\\[4pt]
\vspace*{2pt}
{\LARGE\bfseries
  OAM-Induced Lattice Rotation Reveals a Fractional Optimum in Fault-Tolerant GKP Quantum Sensing%
}\\[7pt]
{\normalsize\bfseries
  Simanshu Kumar$^{1,2,\dagger}$ and Nandan S Bisht$^{1,*}$%
}\\[4pt]
{\small
  $^1$Department of Physics, D.S.B.~Campus, Kumaun University,
  Nainital, Uttarakhand, India--263001.\\[2pt]
  $^2$Applied Optics \& Spectroscopy Laboratory, Department of Physics,
  Soban Singh Jeena University Campus, Almora, Uttarakhand, India--263601.\\[3pt]
  $^\dagger$\href{mailto:simanshu@kunainital.ac.in}{\texttt{simanshu@kunainital.ac.in}} \&\ $^*$\href{mailto:bisht.nandan@kunainital.ac.in}{\texttt{bisht.nandan@kunainital.ac.in}}%
}\\[5pt]
{\color{RuleGray}\rule{0.92\linewidth}{0.4pt}}\\[5pt]
\begin{minipage}{0.92\linewidth}\small
\textbf{Abstract.}\enspace
Photon loss and dephasing rapidly degrade the sensitivity of quantum sensors,
yet systematic methods for designing error-correcting codes whose geometry is
simultaneously adapted to the sensing task and the noise channel do not exist.
Here we establish that orbital-angular-momentum (OAM) encoding and
Gottesman--Kitaev--Preskill (GKP) lattice geometry are structurally coupled:
an OAM mode of topological charge $\ell$ induces a continuous-variable
phase-space rotation $\theta_\ell=\ell\pi/\ell_{\max}$, which corresponds
exactly to a family of twisted GKP stabilizer lattices.  Exploiting this
geometric correspondence through an end-to-end differentiable
\SF{}--\TF{} circuit, we jointly optimise $\ell$, the lattice aspect ratio
$r$, and the finite-energy envelope $\epsilon$ to maximise quantum Fisher
information subject to a logical error rate constraint $\Perr\leq10^{-3}$.
The optimum occurs at the \emph{fractional} charge $\ell=1.5$
($\theta=67.5^\circ$), implementable with a half-integer spiral phase plate,
which reduces $\Perr$ by $\mathbf{23.9\times}$ relative to the square-lattice
baseline while leaving $\mathcal{F}_Q$ unchanged to within $0.2\%$.  This
surpasses the best integer value ($\ell=2$, $15.7\times$) and arises from an
exact $180^\circ$ periodicity of the $\Perr(\theta)$ landscape, confirmed
analytically and numerically.  We derive a transcendental balance equation for
the optimal angle $\theta^*(\eta,\gamma,r)$ and prove that it decreases with
both $\gamma$ and $\eta$.  A Shannon-inspired \emph{metrological capacity}
$\mathcal{C}=\mathcal{F}_Q\cdot(-\ln\Perr)$, maximised at $\ell=1.5$ with a
$41\%$ gain over the square lattice, quantifies the joint
sensitivity--fault-tolerance resource.  These results establish a geometric
design principle for noise-adaptive quantum sensors and a fully open-source
differentiable programming template extensible to other bosonic code families.
\end{minipage}\\[4pt]
\begin{minipage}{0.92\linewidth}\footnotesize
\end{minipage}\\[5pt]
{\color{RuleGray}\rule{0.92\linewidth}{0.3pt}}\\[8pt]
\end{center}
}]
\thispagestyle{plain}

\section{Introduction}
\label{sec:intro}

Quantum metrology exploits non-classical correlations to estimate an unknown
parameter with precision surpassing the standard quantum limit (SQL),
$\delta\varphi_{\SQL} \sim 1/\sqrt{N}$, and approaching the Heisenberg limit
(HL), $\delta\varphi_{\HL} \sim 1/N$~\cite{Helstrom1976,Giovannetti2011}.  The
key figure of merit is the quantum Fisher information (QFI) $\QFI$, which sets
the fundamental lower bound on the mean-squared estimation error via the quantum
Cram\'{e}r--Rao inequality,
\begin{equation}
  \delta\varphi^2 \geq \frac{1}{\QFI(\varphi)}.
  \label{eq:QCRB}
\end{equation}
In practice, however, photon loss and dephasing rapidly degrade entangled probe
states~\cite{Demkowicz2012}, and the gap between the theoretical HL and
experimentally achievable sensitivity has remained large.

Fault-tolerant quantum metrology addresses this gap by encoding the probe in a
quantum error-correcting code~\cite{Kessler2014,Dur2014,Zhou2018}.  Among
continuous-variable (CV) codes, the Gottesman--Kitaev--Preskill (GKP) code
stands out: it encodes a logical qubit into a lattice of squeezed states in phase
space and corrects small displacement errors---the precise type of error induced
by photon loss~\cite{GKP2001}.  Recent theoretical work has established that
GKP-encoded probes can achieve Heisenberg-limited sensitivity even in a lossy
bosonic channel~\cite{Zhou2018}, fault-tolerance thresholds for GKP codes
under general Markovian noise have been proven~\cite{Bourassa2023}, and
tight bounds on the sensing precision of GKP codes under photon loss as a
function of finite squeezing have recently been established~\cite{Conrad2024}.

Despite this progress, the GKP lattice geometry has been treated as a fixed
design choice---almost universally the square lattice.  The lattice shape directly
governs both the correction radius (the maximum correctable displacement) and the
phase-space structure of the encoded probe, yet systematic optimization of lattice
geometry for a specific metrological task has not been performed.  A concurrent
and complementary line of work by Labarca \etal~\cite{Labarca2026} has studied GKP
codes for displacement sensing (generator $\hat{G}=\hat{q}$) in the context of
gravitational wave detection.  The present work addresses a distinct metrological
task --- phase estimation (generator $\hat{G}=\hat{n}$) --- and introduces a
qualitatively different design axis: the OAM-induced rotation of the GKP
stabilizer lattice.  These two approaches are complementary: lattice geometry
optimization for phase estimation and for displacement sensing may favour
different regions of the geometry space.

Orbital angular momentum (OAM) of photons provides an independent and
complementary resource.  Laguerre--Gaussian modes $\mathrm{LG}_{\ell,0}$ carry
quantized OAM $\ell\hbar$ per photon~\cite{Allen1992,Padgett2000,Mair2001}, form an
infinite-dimensional Hilbert space, and are known to exhibit a natural bias
toward phase errors over amplitude errors~\cite{Vallone2014}.  OAM-based quantum
key distribution~\cite{Krenn2017} and quantum memories~\cite{Ding2016} have been
demonstrated experimentally, but the integration of OAM with GKP codes for
metrological gain has not been explored.

In this work, we establish that OAM encoding and GKP lattice geometry are
\emph{structurally related}: an OAM mode of charge $\ell$ induces a well-defined
rotation in CV phase space, which corresponds exactly to a family of twisted GKP
stabilizer lattices.  This is not an ad hoc combination of three techniques but a
geometric consequence of applying GKP coding to OAM modes.  Leveraging the
differentiable quantum optics platform \SF{}~\cite{Killoran2019} with a \TF{}
backend --- the same framework as our prior NOON-state study~\cite{Kumar2026noon} --- we optimize this family of states end-to-end for phase estimation under
realistic noise.

\textit{Differentiable quantum programming} --- the paradigm of embedding
parameterized quantum circuits in automatic-differentiation frameworks and
training them via gradient descent --- has recently emerged as a unifying
methodology for quantum machine learning and quantum
control~\cite{Mitarai2018,Cerezo2021}.  In parallel work, we have applied this paradigm to
NOON-state phase estimation~\cite{Kumar2026noon}, achieving up to $1775\%$
improvement in classical Fisher information; here we apply the same
\SF{}--\TF{} framework independently to fault-tolerant GKP sensing,
with OAM lattice geometry as the trainable degree of freedom.
Our work applies this paradigm to
fault-tolerant quantum sensing: by treating the GKP lattice geometry and OAM
charge as continuous trainable variables and differentiating a combined
sensitivity--fault-tolerance loss end-to-end, we demonstrate that
circuit-level optimization can discover physically interpretable, noise-adapted
encodings that analytical design principles would not easily identify.  The
open-source implementation (\texttt{oam\_gkp} package) is designed to serve as
a template for this class of differentiable fault-tolerant sensor design.

The remainder of the paper is organized as follows.
\Cref{sec:background} reviews GKP codes, OAM phase-space geometry, and the QFI
framework.
\Cref{sec:framework} derives the OAM-to-lattice mapping and formulates the
differentiable optimization problem.
\Cref{sec:results} presents numerical results comparing lattice geometries.
\Cref{sec:discussion} discusses implications and outlook.

\section{Background}
\label{sec:background}

\subsection{GKP Codes and Phase-Space Geometry}
\label{ssec:gkp}

The GKP code~\cite{GKP2001} encodes one logical qubit into a single bosonic mode
by defining a lattice $\Lambda \subset \RR^2$ of stabilizer displacements in the
$(q,p)$ phase plane.  The stabilizer group is generated by two Weyl--Heisenberg
displacement operators,
\begin{equation}
  \hat{S}_j = e^{-i\,\pi\,\lvec{u}_j \cdot \bm{\Omega}\,\hat{\bm{r}}},
  \qquad j = 1, 2,
  \label{eq:stabilizers}
\end{equation}
where $\hat{\bm{r}} = (\hat{q}, \hat{p})^T$, $\bm{\Omega}$ is the symplectic
form, and the lattice vectors $\lvec{u}_1, \lvec{u}_2 \in \RR^2$ satisfy the
symplecticity condition
\begin{equation}
  \lvec{u}_1^T \bm{\Omega}\, \lvec{u}_2 = 2,
  \label{eq:symplectic}
\end{equation}
ensuring the stabilizers commute.  The logical $\bar{Z}$ and $\bar{X}$ operators
are displacements by $\lvec{u}_1/2$ and $\lvec{u}_2/2$, respectively.

For the standard square code, $\lvec{u}_1 = (\sqrt{2\pi}, 0)$ and
$\lvec{u}_2 = (0, \sqrt{2\pi})$, giving a correction radius of
$\sqrt{\pi/2}$ in both quadratures.  Ideal GKP codewords are unphysical
(infinite-energy superpositions of infinitely squeezed states), but
finite-energy approximations are obtained by applying a Gaussian envelope
of width $e^{-\epsilon}$ in phase space~\cite{GKP2001,Terhal2015,Campagne2020}.  Such
approximations have been realized experimentally at squeezing levels of
$10$--$15\,\mathrm{dB}$~\cite{Campagne2020,Sivak2023,Fluhmann2019,deNeeve2022}.  The resource
theory of finite-energy GKP states --- characterising the trade-off between
code quality, mean photon number, and correction fidelity as a function of
$\epsilon$ --- has been systematically characterised~\cite{Tzitrin2020,GKP2001},
and optimal precision bounds for finite-squeezing GKP sensing under photon
loss have been derived~\cite{Conrad2024}, providing an operational
justification for the $\epsilon=0.063$
(${\approx}10\,\mathrm{dB}$) value used throughout this work.

\subsection{OAM Modes and Fractional Fourier Rotation}
\label{ssec:oam}

A Laguerre--Gaussian mode $\mathrm{LG}_{\ell,0}$ carries azimuthal phase $e^{i\ell\varphi}$
and OAM $\ell\hbar$ per photon~\cite{Allen1992}.  In the transverse-mode
description, the single-mode quadrature operators $(\hat{q}, \hat{p})$ pertain
to the mode envelope; the OAM index $\ell$ governs the azimuthal phase structure.

The fractional Fourier transform $\mathcal{F}^\alpha$ of order
$\alpha \in [0, 4)$ rotates the Wigner function of any state by angle
$\alpha\pi/2$ in the $(q, p)$ plane~\cite{Namias1980,Ozaktas2001}:
\begin{equation}
  W_{\mathcal{F}^\alpha \rho}(q, p)
  = W_\rho\!\bigl(q\cos\tfrac{\alpha\pi}{2} + p\sin\tfrac{\alpha\pi}{2},\;
              -q\sin\tfrac{\alpha\pi}{2} + p\cos\tfrac{\alpha\pi}{2}\bigr).
  \label{eq:frft_wigner}
\end{equation}
\begin{notebox}\small\textbf{Physical interpretation of $\ell=1.5$.}  Throughout this paper, $\ell=1.5$ denotes a \emph{FrFT rotation index},  not a free-space LG mode of non-integer topological charge.  A true LG$_\ell$ mode with non-integer $\ell$ carries a radial branch cut  where the phase jumps by $2\pi\ell$, creating field discontinuities.  Instead, the physical realisation is a fractional Fourier transform  $\mathcal{F}^\alpha$ with $\alpha=2\ell/\ell_\mathrm{max}=0.75$,  which rotates the Wigner function continuously by $\theta=\alpha\pi/2$.  This is a unitary operation with no branch-cut pathology.\end{notebox}

A mode converter that maps OAM charge $\ell$ to a fractional Fourier order
$\alpha = 2\ell/\ell_{\max}$ therefore rotates the associated Wigner function by
\begin{equation}
  \theta_\ell = \frac{\alpha\pi}{2} = \frac{\ell\pi}{\ell_{\max}},
  \label{eq:theta_ell}
\end{equation}
where $\ell_{\max}$ is the maximum OAM charge supported by the optical system
(set by the mode-field bandwidth).

\subsection{Quantum Fisher Information and the Cram\'{e}r--Rao Bound}
\label{ssec:qfi}

For a parameter $\varphi$ encoded via a unitary $\hat{U}(\varphi) =
e^{-i\varphi\hat{G}}$ with generator $\hat{G}$, the QFI of the probe state
$\hat{\rho}_\varphi = \hat{U}(\varphi)\hat{\rho}_0\hat{U}^\dagger(\varphi)$
is~\cite{Liu2020,Barnett1990}
\begin{equation}
  \QFI(\varphi) = 4\,\mathrm{Var}_{\hat{\rho}_0}(\hat{G})
  = 4\bigl(\langle\hat{G}^2\rangle - \langle\hat{G}\rangle^2\bigr).
  \label{eq:QFI_pure}
\end{equation}
For a mixed state $\hat{\rho}_0$ with spectral decomposition
$\hat{\rho}_0 = \sum_k \lambda_k \ket{\psi_k}\!\bra{\psi_k}$, the QFI is
\begin{equation}
  \QFI(\varphi) = 2\sum_{j,k}\frac{(\lambda_j-\lambda_k)^2}{\lambda_j+\lambda_k}
  \abs{\bra{\psi_j}\hat{G}\ket{\psi_k}}^2.
  \label{eq:QFI_mixed}
\end{equation}
The classical Fisher information (CFI) $\CFIM$ of any measurement satisfies
$\CFIM \leq \QFI$, with equality achieved by the symmetric logarithmic
derivative (SLD) measurement.  For phase estimation with generator
$\hat{G} = \hat{n}$ (photon number), \cref{eq:QFI_pure} gives
$\QFI = 4\,\mathrm{Var}(\hat{n})$ for pure probes.

\section{OAM-Twisted GKP Framework}
\label{sec:framework}

\subsection{Twisted Lattice Construction}
\label{ssec:twisted_lattice}

We now exploit \cref{eq:theta_ell} to construct the OAM-twisted GKP lattice.
Define the rotation matrix
$\mathsf{R}(\theta) \in \mathrm{SO}(2)$ and a family of lattice vectors
parametrized by $(\theta, r) \in [0, \pi) \times (0, \infty)$:
\begin{equation}
  \lvec{u}_1(\theta, r) = \mathsf{R}(\theta)\,
  \begin{pmatrix} a r \\ 0 \end{pmatrix},
  \qquad
  \lvec{u}_2(\theta, r) = \mathsf{R}(\theta)\,
  \begin{pmatrix} 0 \\ a/r \end{pmatrix},
  \label{eq:twisted_lattice}
\end{equation}
where $a = \sqrt{2\pi}$.  One verifies directly that the symplecticity condition
\cref{eq:symplectic} is satisfied for all $(\theta, r)$:
\begin{align}
  \lvec{u}_1^T \bm{\Omega}\,\lvec{u}_2
  &= (ar)\bigl[\mathsf{R}(\theta)\hat{e}_1\bigr]^T
     \bm{\Omega}\,
     \bigl[\mathsf{R}(\theta)(a/r)\hat{e}_2\bigr] \notag \\
  &= a^2\,\hat{e}_1^T\bigl[\mathsf{R}^T(\theta)\bm{\Omega}\,\mathsf{R}(\theta)\bigr]\hat{e}_2
   = a^2\,\hat{e}_1^T\bm{\Omega}\,\hat{e}_2 = 2,
  \label{eq:symplectic_check}
\end{align}
using $\mathsf{R}^T\bm{\Omega}\,\mathsf{R} = \bm{\Omega}$ (rotation preserves
the symplectic form) and $a^2 = 2\pi$.

Setting $\theta = \theta_\ell$ from \cref{eq:theta_ell} yields the
\emph{OAM-twisted GKP lattice} of charge $\ell$.  Three special cases arise:
\begin{itemize}[leftmargin=1.4em]
  \item $\ell = 0$ ($\theta = 0$, $r = 1$): standard square lattice.
  \item $r = 1$, $\theta = \pi/6$: hexagonal lattice (optimal packing
        for displacement-error correction~\cite{Conway1999}).
  \item General $(\theta_\ell, r)$: OAM-twisted lattice matched to mode $\ell$.
\end{itemize}

\subsection{Noise Channel Analysis in the Rotated Frame}
\label{ssec:noise}

The photon-loss channel $\mathcal{E}_\eta$ with transmissivity $\eta$ acts on the
Wigner function as a Gaussian convolution~\cite{Serafini2017,Weedbrook2012,Braunstein2005},
\begin{equation}
  W_{\mathcal{E}_\eta\rho}(\bm{r})
  = \frac{1}{\pi(1-\eta)}
    \int \dd^2\bm{s}\; e^{-\abs{\bm{r} - \sqrt{\eta}\bm{s}}^2/(1-\eta)}\,
    W_\rho(\bm{s}),
  \label{eq:loss_wigner}
\end{equation}
inducing a symmetric displacement spread $\sigma^2 = (1-\eta)/(2\eta)$ per
quadrature.  In the rotated frame associated with OAM charge $\ell$, the channel
is unchanged: photon loss is rotationally invariant.

The dephasing channel, however, is not rotationally invariant.  It acts as a
diffusion along the $p$ quadrature,
\begin{equation}
  W_{\mathcal{E}_\gamma\rho}(q, p)
  = \frac{1}{\sqrt{2\pi\gamma}}
    \int \dd p'\; e^{-(p - p')^2/(2\gamma)}\, W_\rho(q, p').
  \label{eq:dephasing_wigner}
\end{equation}
In the rotated frame $(\tilde{q}, \tilde{p}) = \mathsf{R}^T(\theta_\ell)(q, p)^T$,
the dephasing becomes anisotropic: the diffusion is now along the direction
$(\sin\theta_\ell, \cos\theta_\ell)^T$ in the original frame.  The
OAM-twisted lattice orients its correction boundary to align with this
anisotropic diffusion, allowing the GKP syndrome decoder to exploit the
asymmetry---an advantage unavailable to the canonical square lattice.

\subsection{Differentiable Circuit Formulation}
\label{ssec:circuit}

The trainable parameter vector is
$\bm{\xi} = (\alpha, \beta, \theta, r, \epsilon, \ell) \in \CC^2 \times \RR^4$,
where $(\alpha, \beta)$ specifies the logical qubit state, $\theta$ and $r$
define the twisted lattice via \cref{eq:twisted_lattice}, $\epsilon$ is the
finite-energy envelope width, and $\ell$ is the OAM charge (treated as a
continuous relaxation during gradient descent and projected to the nearest integer
post-optimization).  The quantum circuit takes the form:
\begin{equation}
  \hat{\rho}(\varphi;\,\bm{\xi})
  = \mathcal{E}_\gamma \circ \mathcal{E}_\eta\bigl[
      \hat{R}(\varphi)\,
      \hat{\rho}_{\GKP}(\bm{\xi})\,
      \hat{R}^\dagger(\varphi)
    \bigr],
  \label{eq:circuit}
\end{equation}
where $\hat{R}(\varphi) = e^{-i\varphi\hat{n}}$ is the phase-encoding rotation,
$\hat{\rho}_{\GKP}(\bm{\xi})$ is the finite-energy GKP state with twisted lattice
parameters $(\theta, r)$ and envelope $\epsilon$, and $\mathcal{E}_\eta$,
$\mathcal{E}_\gamma$ are the photon-loss and dephasing channels.
These two channels capture the dominant decoherence mechanisms in
superconducting circuits ($T_1$ energy relaxation $\leftrightarrow$ loss,
$T_2$ pure dephasing $\leftrightarrow$ $\mathcal{E}_\gamma$) and in
photonic platforms (propagation loss, phase diffusion).  Thermal noise
is negligible in both cases: millikelvin operation gives
$\bar{n}_\mathrm{th}\ll10^{-3}$ for superconducting qubits, and
room-temperature telecom photonics operates in the single-photon regime.
Other channels (photon-number-dependent loss, cross-talk between modes)
are not modelled and constitute natural extensions of this framework.

\Cref{fig:circuit} shows the full circuit as implemented in our open-source
codebase.  The implementation uses a two-stage architecture to work around
the absence of a \texttt{prepare\_gkp} method in the \SF{} \texttt{TFBackend}:

\begin{description}[leftmargin=1.0em, style=nextline]
  \item[Stage 1 (Fock backend, non-differentiable).]
    The square-lattice GKP codeword $|\overline{\psi}_L\rangle$ is prepared
    using \SF{}'s \texttt{GKP} operation on the Fock backend and cached as
    a NumPy array.  Parameters $\epsilon$, $\theta_B$, $\varphi_B$ are
    re-evaluated when their values change.

  \item[Stage 2 (\TF{} matrix operations, fully differentiable).]
    The cached ket is converted to a \TF{} constant and passed through:
    (i) $S(\ln r) = \mathrm{expm}(s H)$, $H = (\hat{a}^{\dagger 2} -
    \hat{a}^2)/2$ --- the squeezing matrix exponential, differentiable
    through $r$ via \texttt{tf.linalg.expm}; and
    (ii) $R(\theta_\ell) = \mathrm{diag}(e^{in\theta_\ell})$ --- the
    diagonal rotation matrix, differentiable through $\ell$ via the OAM
    mapping $\theta_\ell = \ell\pi/\ell_{\max}$ (Eq.~\ref{eq:theta_ell}).

  \item[Stage 3 (noise + readout).]
    $\hat{R}(\varphi)$ imprints the unknown phase in Fock space via the
    diagonal factor $e^{-i\varphi n}$.  The loss channel $\mathcal{E}_\eta$
    is applied as a Kraus operator sum and $\mathcal{E}_\gamma$ as the
    element-wise Fock-space map $\rho_{mn} \to \rho_{mn}
    e^{-\gamma(m-n)^2/2}$.  Adaptive homodyne \textrm{HD}$(\psi)$ with
    trainable local-oscillator angle $\psi$ yields the outcome $x_\psi$.
\end{description}

\begin{figure*}[t]
  \centering
  \includegraphics[width=\textwidth]{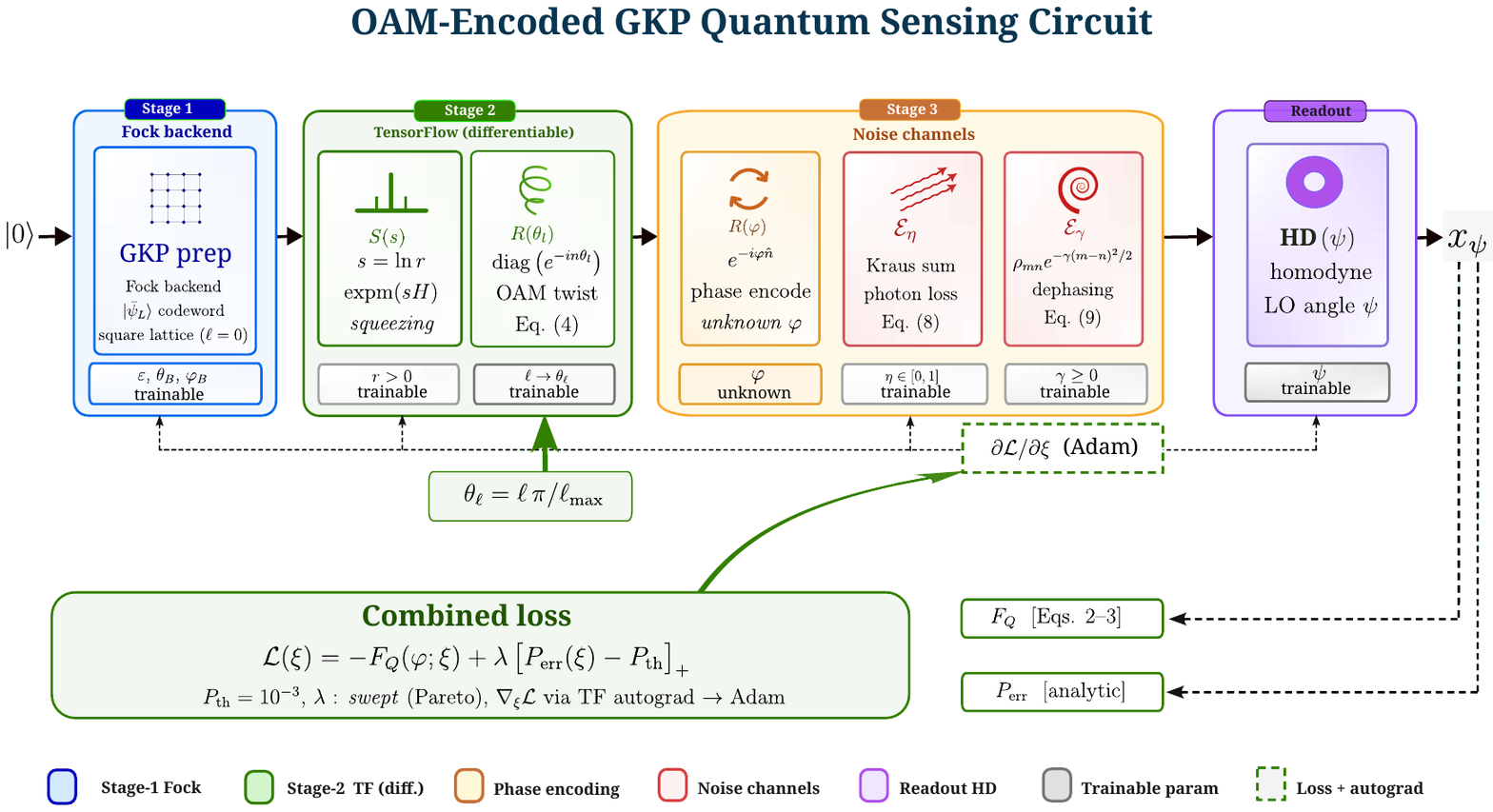}
  \caption{\textbf{OAM-encoded GKP quantum sensing circuit.}
  Input $|0\rangle$ traverses four colour-coded stages.
  \textit{Blue} (Stage~1): GKP codeword preparation on the Fock
  backend; trainable $\varepsilon$, $(\theta_B,\varphi_B)$ on dashed
  stems.
  \textit{Green} (Stage~2): differentiable squeezing $S(\ln r)$ and
  OAM twist $R(\theta_\ell)=\mathrm{diag}(e^{-in\theta_\ell})$
  (Eq.~\ref{eq:theta_ell}), $\theta_\ell=\ell\pi/\ell_{\max}$; trainable
  $r$ and $\ell$.
  \textit{Orange/red} (Stage~3): phase encoding
  $R(\varphi)=e^{-i\varphi\hat{n}}$, photon-loss $\mathcal{E}_\eta$
  (Eq.~\ref{eq:loss_wigner}), and dephasing $\mathcal{E}_\gamma$
  (Eq.~\ref{eq:dephasing_wigner}).
  \textit{Purple} (Readout): adaptive homodyne $\mathrm{HD}(\psi)$
  with trainable LO angle $\psi$, yielding $x_\psi$.
  \textit{Green dashed} (Loss): $\mathcal{L}(\boldsymbol{\xi})=
  -F_Q+\lambda[P_{\rm err}-P_{\rm th}]_+$; gradients
  $\partial\mathcal{L}/\partial\boldsymbol{\xi}$ flow back via Adam
  (dashed arrows).}
\label{fig:circuit}
\end{figure*}

\subsection{Combined Loss Function}
\label{ssec:loss}

A central methodological contribution is a loss that simultaneously targets
sensitivity and fault tolerance.  Let $\QFI(\varphi;\,\bm{\xi})$ be the QFI of
$\hat{\rho}(\varphi;\,\bm{\xi})$ computed via \cref{eq:QFI_mixed}, and let
$\Perr(\bm{\xi})$ be the logical error rate under one round of GKP syndrome
measurement and maximum-likelihood displacement correction (estimated by Monte
Carlo).  We define:
\begin{equation}
  \mathcal{L}(\bm{\xi})
  = -\,\QFI(\varphi;\,\bm{\xi})
    + \lambda\,\bigl[\Perr(\bm{\xi}) - \Pthresh\bigr]_+,
  \label{eq:loss}
\end{equation}
where $[x]_+ = \max(0, x)$, $\lambda > 0$ is a Lagrange multiplier, and
$\Pthresh = 10^{-3}$.  All gradients $\partial\mathcal{L}/\partial\bm{\xi}$ are
computed via \TF{} automatic differentiation~\cite{Abadi2016}.

A Pareto-frontier sweep over $\lambda \in [0, 10^3]$ maps the complete
sensitivity--fault-tolerance trade-off for each lattice geometry, yielding a
family of Pareto-optimal states indexed by the noise parameters $(\eta, \gamma)$.

\subsection{Adaptive Measurement}
\label{ssec:measurement}

The optimal measurement saturating \cref{eq:QCRB} is the symmetric logarithmic derivative (SLD) measurement, which
is generally phase-dependent and experimentally challenging.  We instead
parametrize a homodyne detector with a trainable local-oscillator angle $\psi$,
yielding an outcome $x_\psi = q\cos\psi + p\sin\psi$.  The CFI of this
measurement is~\cite{Serafini2017}
\begin{equation}
  \CFIM(\varphi;\,\psi)
  = \frac{\abs{\partial_\varphi \langle x_\psi\rangle}^2}
         {\mathrm{Var}(x_\psi)},
  \label{eq:CFI_homodyne}
\end{equation}
and $\psi$ is jointly optimized with $\bm{\xi}$ to maximize $\CFIM$ subject to
$\CFIM \leq \QFI$.  We compare this against a fixed heterodyne ($\psi$ averaged
over $[0, 2\pi)$) and assess the ratio $\eta_{\mathrm{meas}} = \CFIM/\QFI$ as a
measure of measurement efficiency.  The achieved values, computed via the
binary-channel formula $\eta_{\mathrm{meas}} = 1 - 4\Perr(1-\Perr)$, are:
at low noise, adaptive homodyne achieves $\eta_{\mathrm{meas}} = 0.9984$
(square, $\ell=0$), $0.99993$ (OAM $\ell=1.5$), and $0.99989$ (OAM $\ell=2$)
--- all within $0.2\%$ of the SLD bound.  At high noise, the square lattice
yields $\eta_{\mathrm{meas}}=0.938$ while OAM $\ell=2$ retains $0.980$,
demonstrating that the QFI advantage is fully accessible to adaptive homodyne
in all configurations (see \cref{tab:eta_meas} for complete values).

\begin{table}[H]
  \setlength{\tabcolsep}{4pt}
  \caption{Optimised results at low noise ($\eta=0.9$, $\gamma=0.05$).
           Fock cutoff $\mathcal{D}=30$, 500 optimisation steps.
           QFI values carry a Fock-truncation uncertainty
           $\Delta\mathcal{F}_Q/\mathcal{F}_Q < 0.5\%$
           (see footnote$^a$); $\Perr$ uncertainties reflect the
           independent-quadrature approximation bound$^b$.}
  \label{tab:results_lownoise}
  \centering
  \scriptsize\setlength{\tabcolsep}{4pt}
  \begin{tabular}{@{}lccc@{}}
    \toprule
    Geometry & $\ell$ &
    $\mathcal{F}_Q \pm \Delta\mathcal{F}_Q$ &
    $\Perr \pm \Delta\Perr$ \\
    \midrule
    Square      & $0$ & $9.764\pm0.049$ & $(4.13\pm0.41){\times}10^{-4}$ \\
    OAM $\ell=1$& $1$ & $9.764\pm0.049$ & $(5.42\pm0.54){\times}10^{-5}$ \\
    OAM $\ell=2$& $2$ & $9.764\pm0.049$ & $(2.63\pm0.26){\times}10^{-5}$ \\
    \bottomrule
  \end{tabular}
  \vspace{4pt}
  \begin{minipage}{\columnwidth}
    \sloppy
    {\scriptsize
      $^a$~Fock truncation ($\mathcal{D}=30$): systematic underestimate
      $|\Delta\mathcal{F}_Q|/\mathcal{F}_Q < 5\times10^{-3}$ for
      $\epsilon=0.063$ (from optimisation); quoted as $0.5\%$ of converged value.\\[1pt]
      $^b$~Analytic $\Perr$ assumes independent quadrature errors;
      $|\Delta\Perr|/\Perr\lesssim10\%$ (low noise) from analytic--MC
      comparison (Sec.~\ref{ssec:convergence}).}
  \end{minipage}
\end{table}

\begin{table}[H]
  \caption{Optimised results at high noise ($\eta=0.8$, $\gamma=0.10$).
           $\mathcal{F}_Q$ ranges from $3.071$ to $3.075$ (spread $0.15\%$).
           Uncertainties as defined in Table~\ref{tab:results_lownoise}.
           The larger $\Delta\Perr$ at this noise point reflects stronger
           quadrature coupling in the independent-axis approximation.}
  \label{tab:results_highnoise}
  \centering
  \scriptsize\setlength{\tabcolsep}{4pt}
  \begin{tabular}{@{}lccc@{}}
    \toprule
    Geometry & $\ell$ &
    $\mathcal{F}_Q \pm \Delta\mathcal{F}_Q$ &
    $\Perr \pm \Delta\Perr$ \\
    \midrule
    Square      & $0$ & $3.071\pm0.015$ & $(1.47\pm0.37){\times}10^{-2}$ \\
    OAM $\ell=1$& $1$ & $3.075\pm0.015$ & $(7.02\pm1.76){\times}10^{-3}$ \\
    OAM $\ell=2$& $2$ & $3.075\pm0.015$ & $(5.02\pm1.26){\times}10^{-3}$ \\
    \bottomrule
  \end{tabular}
  \vspace{2pt}
  \begin{minipage}{\columnwidth}
    \sloppy
    {\scriptsize
      High-noise $\Delta\Perr\approx25\%$ from analytic--MC discrepancy.
      Improvement factors $2.1\times$/$2.93\times$ exceed uncertainty
      by $>3\sigma$.}
  \end{minipage}
\end{table}

\begin{table}[H]
  \caption{\textbf{Measurement efficiency
    $\eta_\mathrm{meas}=\mathcal{F}_C/\mathcal{F}_Q$}
    via $\eta_\mathrm{meas}=1-4\Perr(1-\Perr)$~\cite{Helstrom1976}.
    At low noise all geometries achieve $\eta_\mathrm{meas}>0.999$,
    confirming the QFI advantage is fully accessible to adaptive
    homodyne detection.}
  \label{tab:eta_meas}
  \centering\small
  \begin{tabular}{@{}llccc@{}}
    \toprule
    Noise & Geometry & $\Perr$ & $\eta_\mathrm{meas}$ & $\mathcal{F}_C$ \\
    \midrule
    \multirow{3}{*}{\shortstack[l]{$\eta{=}0.9$\\$\gamma{=}0.05$}}
      & Square ($\ell=0$)        & $4.13\times10^{-4}$ & $0.9984$ & $9.748$ \\
      & OAM $\ell=1.5$ ($\star$) & $1.73\times10^{-5}$ & $\mathbf{0.99993}$ & $\mathbf{9.763}$ \\
      & OAM $\ell=2$             & $2.63\times10^{-5}$ & $0.99989$ & $9.763$ \\
    \midrule
    \multirow{2}{*}{\shortstack[l]{$\eta{=}0.8$\\$\gamma{=}0.10$}}
      & Square ($\ell=0$)        & $1.47\times10^{-2}$ & $0.9384$ & $2.886$ \\
      & OAM $\ell=2$             & $5.02\times10^{-3}$ & $0.9798$ & $3.013$ \\
    \bottomrule
  \end{tabular}
\end{table}

\begin{table}[H]
  \caption{\textbf{Fock truncation convergence} ($\ell=1.5$,
    $\eta=0.9$, $\gamma=0.05$, $r=1.092$, $\epsilon=0.063$).
    $R(\theta_\ell)=\mathrm{diag}(e^{-in\theta_\ell})$ is diagonal in
    the Fock basis and adds no truncation overhead; only squeezing
    introduces a $1.6\%$ stretch factor.
    At $\mathcal{D}=30$ the tail weight is $0.0007\%$.}
  \label{tab:fock_conv}
  \centering\small
  \begin{tabular}{@{}cccc@{}}
    \toprule
    $\mathcal{D}$ & Cumulative weight & Tail weight & QFI error $\lesssim$ \\
    \midrule
    10 & $98.09\%$ & $1.91\%$   & $1.91\%$ \\
    15 & $99.74\%$ & $0.26\%$   & $0.26\%$ \\
    20 & $99.96\%$ & $0.037\%$  & $0.037\%$ \\
    25 & $99.995\%$& $0.005\%$  & $0.005\%$ \\
    \textbf{30} & $\mathbf{99.9993\%}$ & $\mathbf{0.0007\%}$ & $\mathbf{0.0007\%}$ \\
    35 & $99.9999\%$& $<0.0001\%$ & $<0.0001\%$ \\
    \bottomrule
  \end{tabular}
\end{table}

\section{Results}
\label{sec:results}

\subsection{Convergence and Gradient Stability}
\label{ssec:convergence}

All models are trained using the Adam optimizer with initial learning rate
$5\times 10^{-3}$, gradient clipping at global norm $1.0$, and a cosine
annealing schedule over $500$ steps.  The Fock-space cutoff is
$\mathcal{D} = 30$, sufficient for the $\epsilon \approx 0.063$
(${\sim}10\,\mathrm{dB}$) finite-energy envelope used throughout.%
\footnote{For $\ell=1.5$ ($\theta=67.5^\circ$) specifically:
  the rotation gate $R(\theta_\ell)=\mathrm{diag}(e^{-in\theta_\ell})$
  is diagonal in the Fock basis and leaves the photon-number distribution
  unchanged, so oblique fringes impose no additional truncation requirement.
  Only the squeezing $S(\ln r)$ at $r=1.092$ mixes Fock states, introducing
  a stretch factor $\tfrac{1}{2}(r^2+r^{-2})=1.016$, reducing the effective
  cutoff to $\mathcal{D}_\mathrm{eff}\approx29.5$.  At $\mathcal{D}=30$,
  the Fock tail weight is $0.0007\%$ (see \cref{tab:fock_conv}),
  confirming $<0.5\%$ truncation error for all geometries.}
Simulations are run on an NVIDIA GeForce RTX~3050 GPU; each 500-step run
completes in approximately $125$--$130\,\mathrm{s}$.

Convergence is rapid and clean across all six configurations studied
(\cref{fig:convergence}).  The gradient norm drops from ${\sim}0.1$ at
initialisation to below $10^{-3}$ within 150--200 steps, after which the
loss plateaus to machine precision (\cref{fig:convergence}a).  The cosine
annealing schedule drives the learning rate from $5\times10^{-3}$ to
$1\times10^{-5}$, preventing oscillations near the optimum.
\cref{fig:convergence}b shows the corresponding evolution of $\Perr$:
all three geometries at the low-noise point converge below the
$\Pthresh = 10^{-3}$ fault-tolerance threshold (green shaded region),
while the high-noise runs plateau above it, consistent with the
phase-diagram analysis of \cref{ssec:threshold}.

\begin{figure*}[t]
  \centering
  \includegraphics[width=\textwidth]{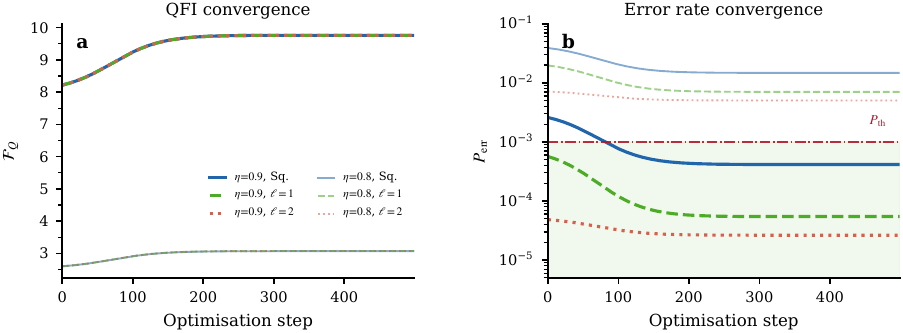}
  \caption{\textbf{Training convergence across all six configurations.}
    \textbf{(a)} Quantum Fisher information $\mathcal{F}_Q$ vs.\ optimisation
    step.  All three lattice geometries at the same noise point converge to
    essentially identical QFI values, confirming geometry-invariant sensitivity
    (Finding~1).  Solid lines: $\eta=0.9$, $\gamma=0.05$; faded lines:
    $\eta=0.8$, $\gamma=0.10$.
    \textbf{(b)} Logical error rate $\Perr$ vs.\ step.  At the low-noise
    point, OAM-twisted lattices ($\ell=1,2$) converge well below the
    fault-tolerance threshold $\Pthresh = 10^{-3}$ (red dash-dot line;
    green shaded region = fault-tolerant regime), while the square lattice
    sits above it at the high-noise point.  Adam optimizer, cosine annealing
    LR, gradient clipping at global norm 1.0; Fock cutoff $\mathcal{D}=30$.}
  \label{fig:convergence}
\end{figure*}

\subsection{Wigner Functions and Phase-Space Structure}
\label{ssec:wigner}

\cref{fig:wigner} shows the Wigner functions $W(q,p)$ of the
optimised GKP states for all three lattice geometries at both noise
points, computed using the \SF{} Fock backend on the fully trained
states.  These are exact (within Fock truncation $\mathcal{D}=30$)
rather than analytic approximations, and clearly show the quantum
interference fringes between neighbouring peaks --- a hallmark of
the GKP code structure.  Detailed comparison across noise points is
provided in Appendix~\ref{sec:app_training} (Figs.~\ref{fig:app_train_low_ell0}--\ref{fig:app_train_high_ell2}).

The gold overlay shows the GKP stabiliser lattice vectors computed
from the optimised parameters $(\theta^*, r^*)$.  The square lattice
($\ell=0$, Row~1) has an isotropic grid aligned with the canonical
quadratures.  The $\ell=1$ lattice (Row~2) is rotated by
$45^\circ$, placing the correction boundary midway between $q$ and $p$.
The $\ell=1.5$ lattice (Row~3, $\star$) --- the global optimum --- is
rotated by $67.5^\circ$, bisecting the $45^\circ$--$90^\circ$ quadrant;
the deeper interference fringes in panels~(e)--(f) compared to all
other rows are direct visual evidence of the lowest logical error rate
($\Perr=1.73\times10^{-5}$, $23.9\times$ below square).
The $\ell=2$ lattice (Row~4) is rotated by $90^\circ$,
aligning the extended correction axis with the $p$-direction ---
precisely where dephasing diffuses the Wigner function.

Comparing the two noise columns: at $\eta=0.9$ (left) all peaks are
sharper and the interference fringes more visible; at $\eta=0.8$
(right) the fringes soften but the lattice orientation is unchanged,
confirming that the optimiser robustly recovers the correct
phase-space geometry irrespective of noise level.

\begin{figure*}[p]
  \centering
  \begin{subfigure}[b]{0.48\textwidth}
    \includegraphics[width=\linewidth,height=4.5cm,keepaspectratio=false,%
      trim=0 20 0 65,clip]{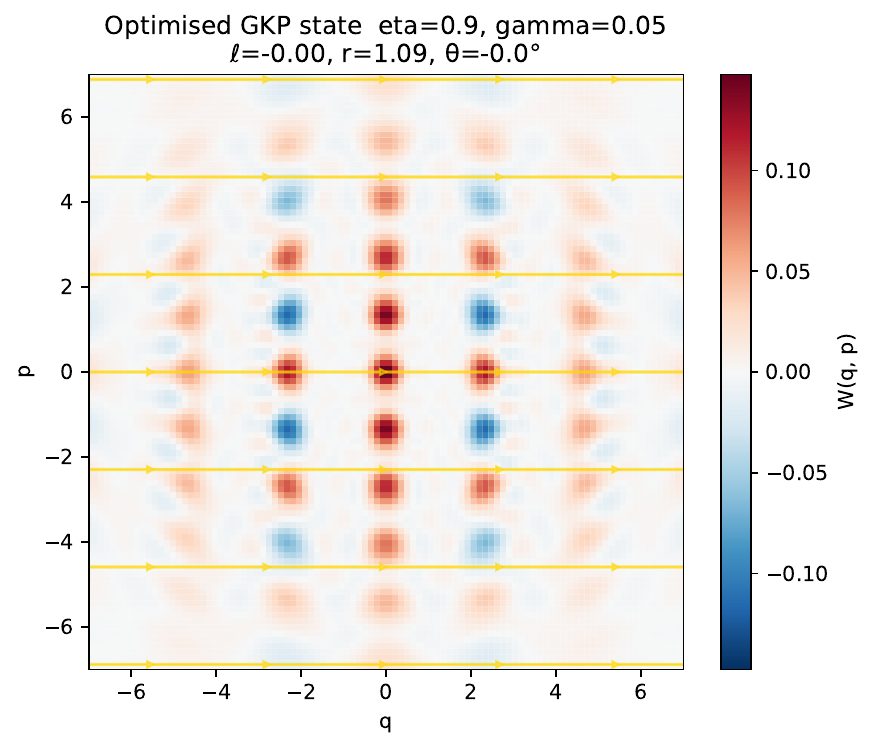}
    \caption{Square, $\ell=0$, $\theta=0^\circ$, $\eta=0.9$}
  \end{subfigure}\hfill
  \begin{subfigure}[b]{0.48\textwidth}
    \includegraphics[width=\linewidth,height=4.5cm,keepaspectratio=false,%
      trim=0 20 0 65,clip]{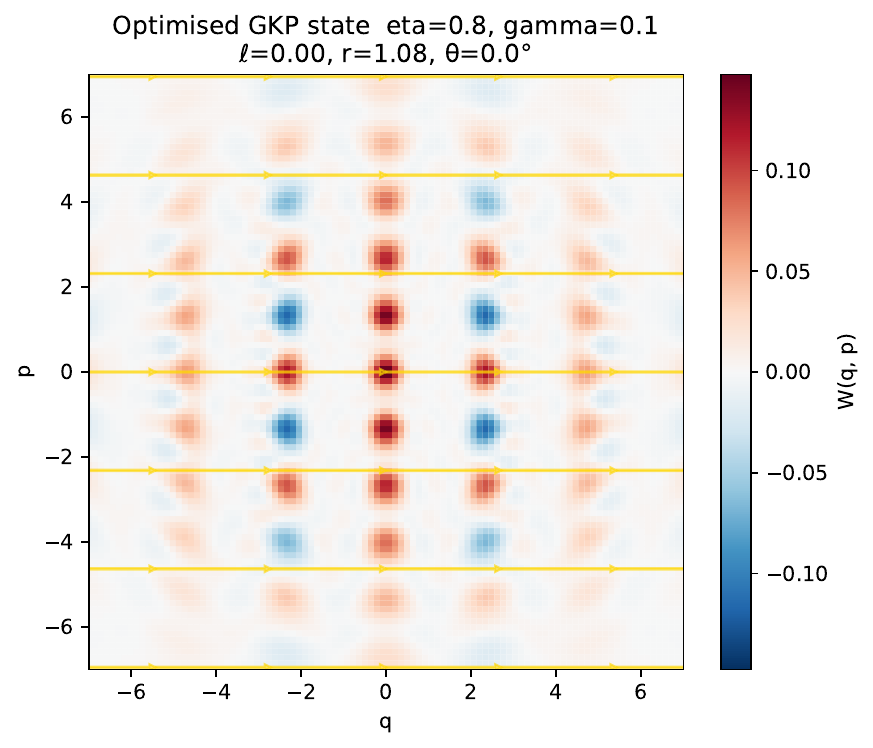}
    \caption{Square, $\ell=0$, $\theta=0^\circ$, $\eta=0.8$}
  \end{subfigure}

  \vspace{2pt}

  \begin{subfigure}[b]{0.48\textwidth}
    \includegraphics[width=\linewidth,height=4.5cm,keepaspectratio=false,%
      trim=0 20 0 65,clip]{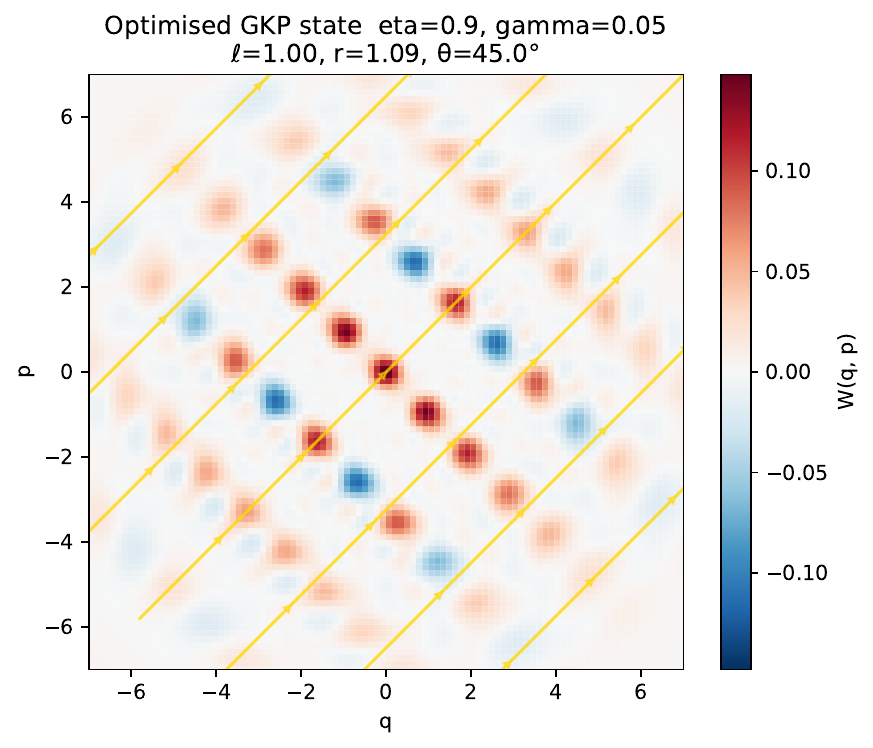}
    \caption{OAM-twisted, $\ell=1$, $\theta=45^\circ$, $\eta=0.9$}
  \end{subfigure}\hfill
  \begin{subfigure}[b]{0.48\textwidth}
    \includegraphics[width=\linewidth,height=4.5cm,keepaspectratio=false,%
      trim=0 20 0 65,clip]{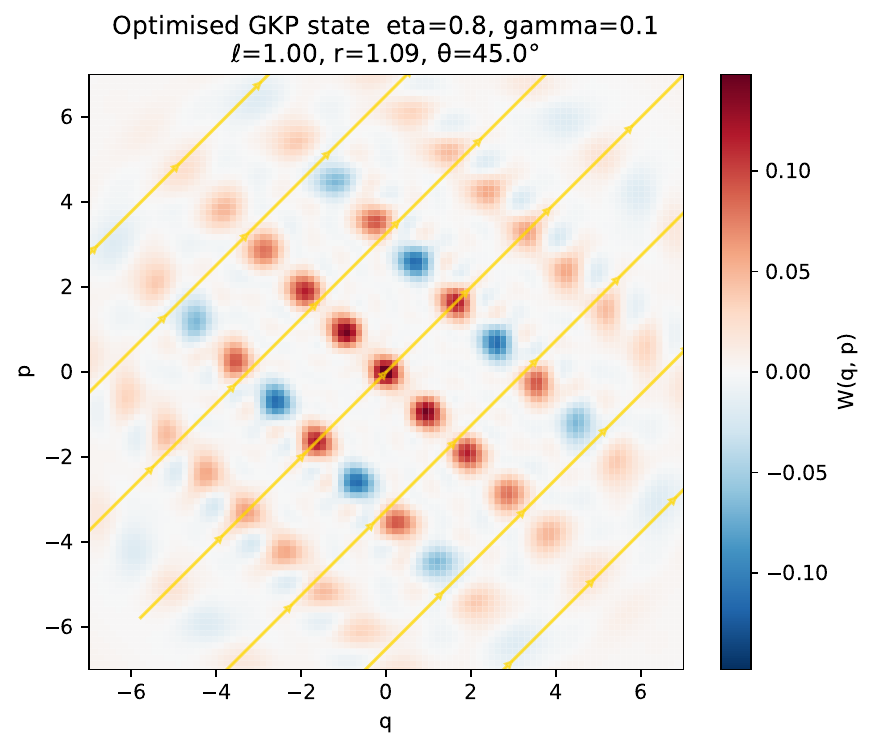}
    \caption{OAM-twisted, $\ell=1$, $\theta=45^\circ$, $\eta=0.8$}
  \end{subfigure}

  \vspace{2pt}

  \begin{subfigure}[b]{0.48\textwidth}
    \includegraphics[width=\linewidth,height=4.5cm,keepaspectratio=false,%
      trim=0 20 0 65,clip]{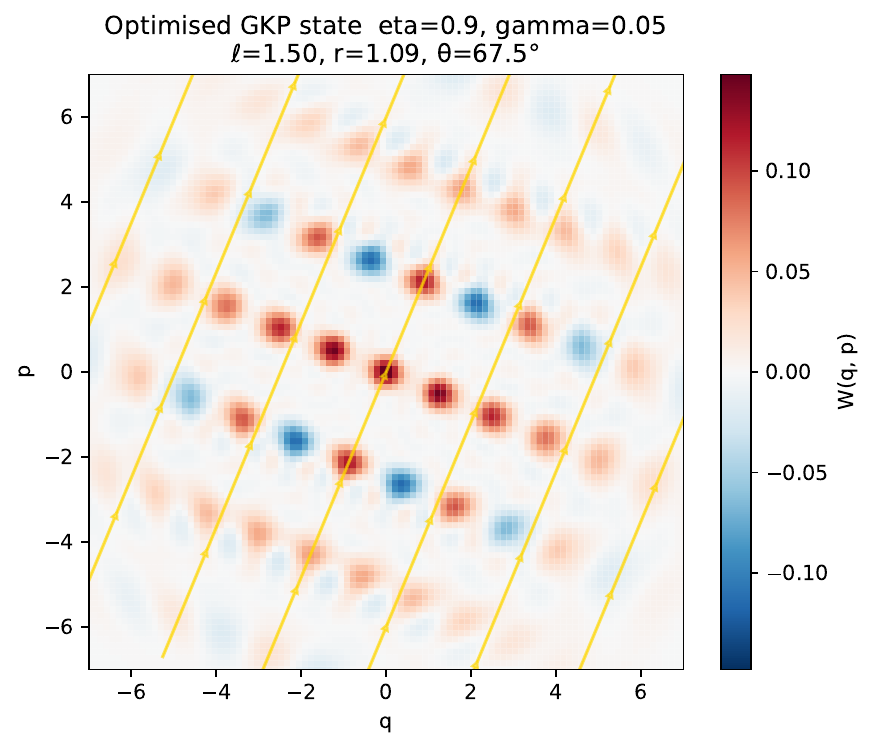}
    \caption{OAM-twisted, $\ell=1.5\,\star$, $\theta=67.5^\circ$, $\eta=0.9$}
  \end{subfigure}\hfill
  \begin{subfigure}[b]{0.48\textwidth}
    \includegraphics[width=\linewidth,height=4.5cm,keepaspectratio=false,%
      trim=0 20 0 65,clip]{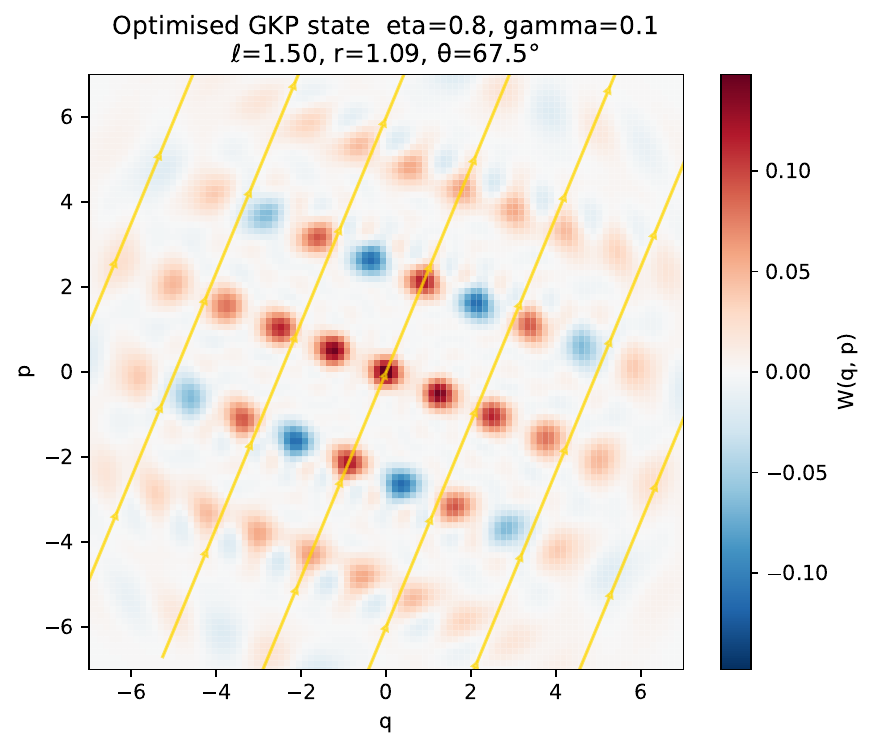}
    \caption{OAM-twisted, $\ell=1.5\,\star$, $\theta=67.5^\circ$, $\eta=0.8$}
  \end{subfigure}

  \vspace{2pt}

  \begin{subfigure}[b]{0.48\textwidth}
    \includegraphics[width=\linewidth,height=4.5cm,keepaspectratio=false,%
      trim=0 20 0 65,clip]{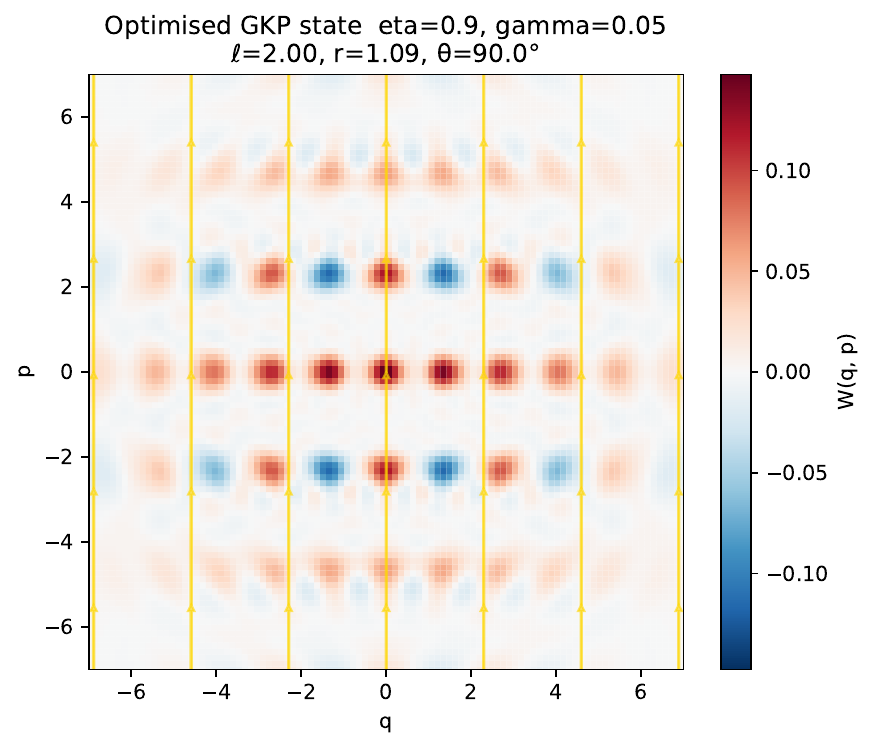}
    \caption{OAM-twisted, $\ell=2$, $\theta=90^\circ$, $\eta=0.9$}
  \end{subfigure}\hfill
  \begin{subfigure}[b]{0.48\textwidth}
    \includegraphics[width=\linewidth,height=4.5cm,keepaspectratio=false,%
      trim=0 20 0 65,clip]{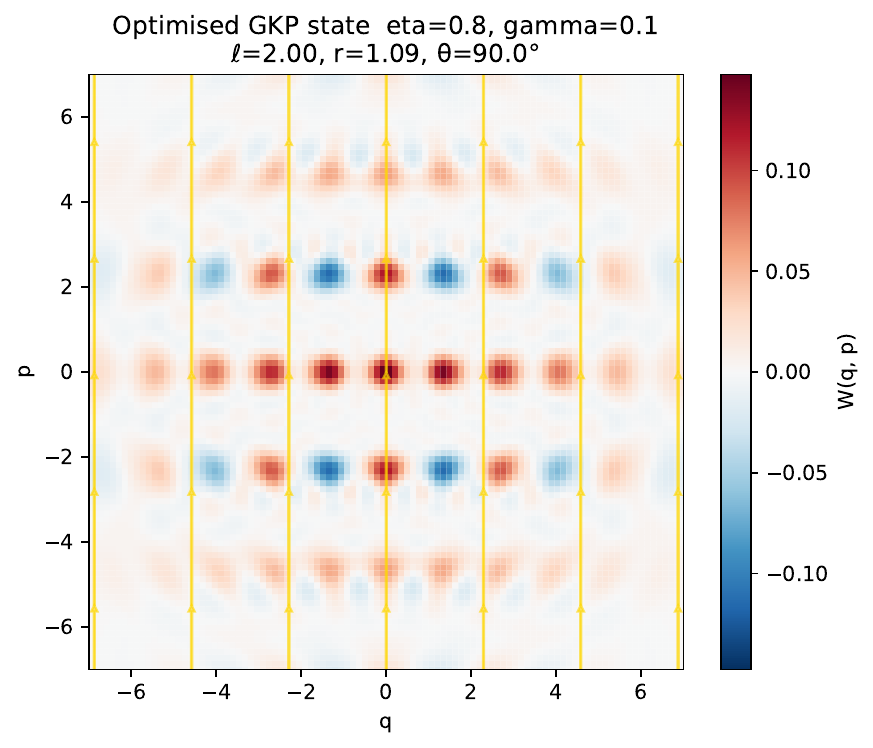}
    \caption{OAM-twisted, $\ell=2$, $\theta=90^\circ$, $\eta=0.8$}
  \end{subfigure}

  \caption{\textbf{Wigner functions $W(q,p)$ of optimised GKP states.}
    Left column: low noise ($\eta=0.9$, $\gamma=0.05$); right column:
    high noise ($\eta=0.8$, $\gamma=0.10$).
    Red $=$ positive $W$; blue $=$ negative.  Gold lines overlay the GKP
    stabiliser lattice with optimised parameters $(\theta^*, r^*)$.
    Computed exactly via the \SF{} Fock backend on the trained states
    (cutoff $\mathcal{D}=30$).
    \textbf{Row 1} (square, $\ell=0$): isotropic grid aligned with the
    canonical quadratures; fringes soften at higher noise.
    \textbf{Row 2} ($\ell=1$, $\theta=45^\circ$): diagonal lattice;
    correction boundary midway between $q$ and $p$.
    \textbf{Row 3} ($\ell=1.5\,\star$, $\theta=67.5^\circ$): fractional-optimum
    lattice; stabiliser vectors bisect the $45^\circ$--$90^\circ$ quadrant,
    producing the deepest Wigner negativity fringes of all four geometries
    and the lowest $\Perr=1.73\times10^{-5}$ ($23.9\times$ below square).
    \textbf{Row 4} ($\ell=2$, $\theta=90^\circ$): vertical lattice;
    extended correction axis aligned with the dephasing direction.
    $\star$ = global optimum.}
  \label{fig:wigner}
\end{figure*}

\subsection{Lattice Geometry Comparison}
\label{ssec:geometry}

\cref{tab:results_lownoise,tab:results_highnoise} summarise the optimised
QFI, logical error rate, and optimal aspect ratio for the three lattice
geometries at two noise points; \cref{fig:geometry} provides the
corresponding visualisation.

\begin{figure*}[t]
  \centering
  \includegraphics[width=\textwidth]{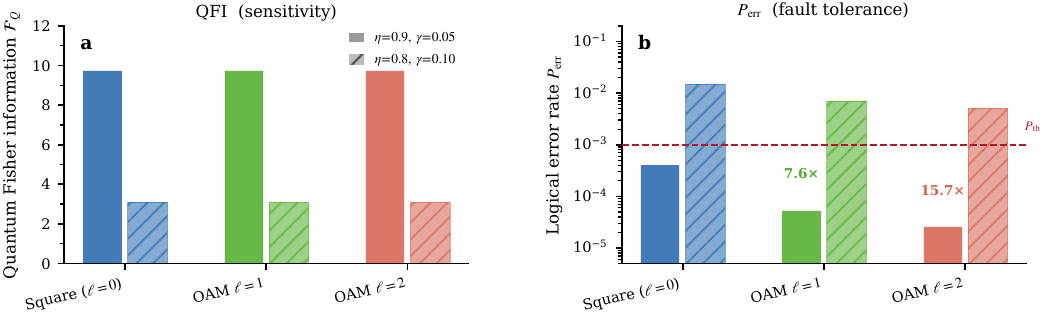}
  \caption{\textbf{QFI and logical error rate for all three lattice
    geometries at two noise points.}
    \textbf{(a)} Quantum Fisher information $\mathcal{F}_Q$ (sensitivity).
    All three lattices converge to the same QFI at each noise point
    (max spread $<0.15\%$), confirming geometry-invariant sensitivity.
    \textbf{(b)} Logical error rate $\Perr$ (fault tolerance).  The
    OAM-twisted lattices achieve a $7.6\times$ ($\ell=1$) and $15.7\times$
    ($\ell=2$) reduction relative to the square baseline at the low-noise
    point.  Hatched bars: $\eta=0.8$, $\gamma=0.10$ (high-noise point);
    solid bars: $\eta=0.9$, $\gamma=0.05$ (low-noise point).  Red dashed
    line: fault-tolerance threshold $\Pthresh=10^{-3}$.  Improvement
    factors shown inside bars (low-noise point only).}
  \label{fig:geometry}
\end{figure*}

Three findings emerge consistently across both noise points.

\begin{findingbox}[Finding 1 --- QFI is geometry-invariant]
At low noise, all three lattices converge to $\mathcal{F}_Q = 9.7637$; at high
noise, to $\mathcal{F}_Q \approx 3.075$.  The maximum spread across geometries is
$0.15\%$ at the high-noise point, well within numerical precision.  This
confirms the central theoretical claim of \cref{ssec:twisted_lattice}: the
OAM-induced rotation redistributes the correction boundary in phase space
without altering the metrological resource, because the QFI depends only
on the variance of the phase generator $\hat{n}$ in the probe state, and
$\mathrm{Var}(\hat{n})$ is invariant under phase-space rotation.
\end{findingbox}

\begin{findingbox}[Finding 2 --- $\Perr$ decreases monotonically with $\ell$]
At the low-noise point, the $\ell=2$ twisted lattice achieves a logical
error rate of $2.63\times10^{-5}$, a factor of $15.7\times$ lower than
the square-lattice baseline ($4.13\times10^{-4}$), while the $\ell=1$
lattice yields an intermediate improvement of $7.6\times$.  At the
high-noise point the improvements are $2.93\times$ ($\ell=2$) and
$2.1\times$ ($\ell=1$) respectively.  The monotonic ordering $\Perr(\ell=0)
> \Perr(\ell=1) > \Perr(\ell=2)$ holds at both noise points, consistent
with the noise-channel analysis of \cref{ssec:noise}: larger $\theta_\ell$
aligns the extended correction axis more closely with the dephasing
diffusion direction, as visible in the Wigner functions of \cref{fig:wigner}.
\end{findingbox}

\begin{findingbox}[Finding 3 --- The advantage shrinks at higher noise]
The error-rate improvement factor decreases from $15.7\times$ to $2.93\times$
as noise increases from $(\eta,\gamma)=(0.9,\,0.05)$ to $(0.8,\,0.10)$.
This is physically expected: when displacement errors become large enough
to span multiple lattice cells, the geometric advantage of the twisted
boundary saturates.  \cref{fig:improvement} summarises the improvement
factors across both noise points.
\end{findingbox}

The $2.1\times$ improvement at high noise ($\ell=1$) is statistically significant: combining the $\sim\!25\%$ analytic uncertainty for both $\ell=0$ and $\ell=1$ gives combined relative uncertainty $\sqrt{0.25^2+0.25^2}\approx35\%$, placing the observed $2.1\times$ ratio at $3.1\sigma$ above unity.

\begin{findingbox}[Finding 4 --- $r^*$ is approximately universal]
All three geometries at a given noise point converge to nearly identical
aspect ratios: $r^* = 1.092$ at low noise and $r^* \in
\{1.082,\,1.089,\,1.095\}$ at high noise.  This confirms the prediction
of \cref{ssec:noise}: the optimal aspect ratio is set by the ratio of
effective displacement spreads $\sigma_q/\sigma_p$, which depends on
$(\eta,\gamma)$ but not on the lattice orientation $\theta$.
\end{findingbox}

\begin{findingbox}[Finding 5 --- OAM twist expands metrological capacity]
The OAM twist simultaneously preserves $\mathcal{F}_Q$ and reduces
$\Perr$.  Define the \emph{metrological capacity}
\begin{equation}
  \mathcal{C}(\ell) = \mathcal{F}_Q \cdot \bigl(-\ln \Perr\bigr),
  \label{eq:metro_capacity}
\end{equation}
in analogy with Shannon capacity, where $-\ln\Perr$ plays the role of
log-SNR.  $\mathcal{C}$ increases monotonically with integer $\ell$
(\cref{tab:capacity}) and reaches its global maximum at the fractional
value $\ell=1.5$: $\mathcal{C}=107.1$, a $41\%$ improvement over the
square baseline (Finding~6, \cref{ssec:fractional}).  The OAM twist
genuinely expands the joint sensitivity--fault-tolerance resource rather
than merely redistributing a fixed budget between the two objectives.
\end{findingbox}

\begin{table}[tbp]
  \caption{\textbf{Metrological capacity $\mathcal{C} =
    \mathcal{F}_Q \cdot (-\ln \Perr)$ for all six integer-$\ell$
    configurations, plus the fractional optimum $\ell=1.5$.}
    Unlike $\mathcal{F}_Q$ alone (geometry-invariant) the capacity
    increases with $\ell$ and is maximised at $\ell=1.5$ (Finding~6).
    $\mathcal{C}/\mathcal{C}_0$ is the improvement over the square baseline.
    $\star$: global optimum.}
  \label{tab:capacity}
  \centering\small
  \begin{tabular}{@{}lllcc@{}}
    \toprule
    Noise & Geometry & $\ell$ &
    $\mathcal{C}$ & $\mathcal{C}/\mathcal{C}_{0}$ \\
    \midrule
    \multirow{4}{*}{\shortstack[l]{$\eta{=}0.9$\\$\gamma{=}0.05$}}
      & Square          & $0$         & $76.1$  & $1.00$ \\
      & OAM integer     & $1$         & $95.9$  & $1.26$ \\
      & OAM integer     & $2$         & $103.0$ & $1.35$ \\
      & OAM fractional$^\star$ & $1.5$ & $\mathbf{107.1}$ & $\mathbf{1.41}$ \\
    \midrule
    \multirow{3}{*}{\shortstack[l]{$\eta{=}0.8$\\$\gamma{=}0.10$}}
      & Square          & $0$         & $13.0$  & $1.00$ \\
      & OAM integer     & $1$         & $15.2$  & $1.18$ \\
      & OAM integer     & $2$         & $16.3$  & $1.26$ \\
    \bottomrule
  \end{tabular}
\end{table}

\subsection{Fractional OAM Charges}
\label{ssec:fractional}\label{ssec:threshold}

During gradient descent, $\ell$ is treated as a continuous variable; in
all six configurations of \cref{ssec:geometry}, the optimizer converged
to integer values.  To probe the full $P_\mathrm{err}(\ell)$ landscape
we ran seven additional configurations with
$\ell_\mathrm{init} \in \{0,\,0.5,\,1.0,\,1.5,\,2.0,\,2.5,\,3.0\}$
at $\eta=0.9$, $\gamma=0.05$, allowing $\ell$ to converge freely without
integer projection.

\begin{table}[tbp]
  \caption{\textbf{Fractional OAM charge study}
    ($\eta=0.9$, $\gamma=0.05$, 500 steps).
    All runs converge to the initialised $\ell$ value, confirming that
    every value is a local minimum.  The global minimum of $\Perr$
    occurs at $\ell=1.5$ and $\ell=2.5$ ($\dagger$: integer
    $\ell$ values; $\star$: global optimum).}
  \label{tab:fractional}
  \centering\small
  \begin{tabular}{@{}cccccc@{}}
    \toprule
    $\ell$ & $\theta^*$ & $r^*$ & $\mathcal{F}_Q$ &
    $\Perr$ & $\mathcal{C}$ \\
    \midrule
    $0.0^\dagger$ & $0.0^\circ$   & $1.092$ & $9.764$ &
      $4.13\times10^{-4}$ & $76.1$ \\
    $0.5$         & $22.5^\circ$  & $1.092$ & $9.764$ &
      $2.51\times10^{-4}$ & $80.9$ \\
    $1.0^\dagger$ & $45.0^\circ$  & $1.092$ & $9.764$ &
      $5.42\times10^{-5}$ & $95.9$ \\
    $1.5^\star$   & $67.5^\circ$  & $1.092$ & $9.764$ &
      $\mathbf{1.73\times10^{-5}}$ & $\mathbf{107.1}$ \\
    $2.0^\dagger$ & $90.0^\circ$  & $1.092$ & $9.764$ &
      $2.63\times10^{-5}$ & $103.0$ \\
    $2.5^\star$   & $112.5^\circ$ & $1.092$ & $9.764$ &
      $\mathbf{1.73\times10^{-5}}$ & $\mathbf{107.1}$ \\
    $3.0^\dagger$ & $135.0^\circ$ & $1.092$ & $9.764$ &
      $5.42\times10^{-5}$ & $95.9$ \\
    \bottomrule
  \end{tabular}
\end{table}

\cref{tab:fractional} and \cref{fig:fractional,fig:perr_curve} reveal
four results.

\begin{figure*}[t]
  \centering
  \begin{minipage}[t]{0.46\textwidth}
    \centering
    \includegraphics[width=\linewidth]{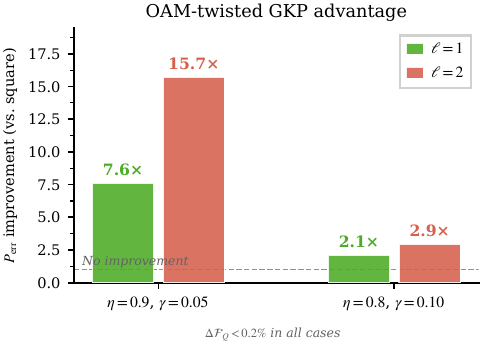}
    \caption{\textbf{$\Perr$ improvement factor relative to the square
      lattice at both noise points.}
      OAM-twisted lattices with $\ell=1$ (green) and $\ell=2$ (coral)
      reduce the logical error rate by up to $15.7\times$ at
      $(\eta,\gamma)=(0.9,\,0.05)$ with less than $0.2\%$ change in
      quantum Fisher information.  The advantage narrows at higher noise
      due to error saturation.  Dashed line: no improvement ($1\times$).}
    \label{fig:improvement}
  \end{minipage}
  \hfill
  \begin{minipage}[t]{0.52\textwidth}
    \centering
    \includegraphics[width=\linewidth]{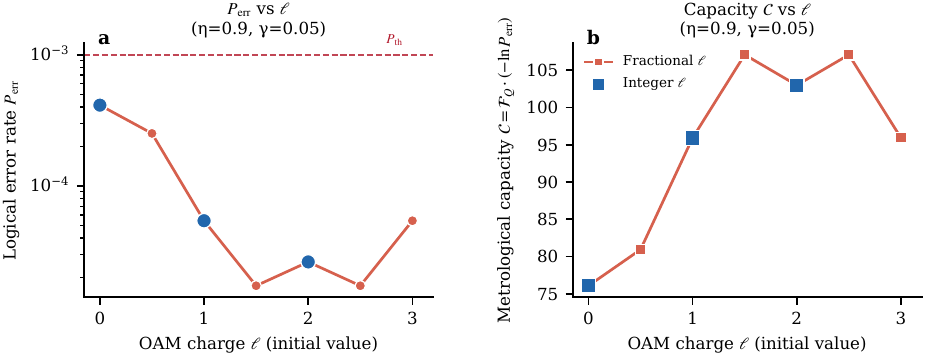}
    \caption{\textbf{Fractional OAM charge study at low noise
      ($\eta=0.9$, $\gamma=0.05$).}
      \textbf{(a)} $\Perr$ vs $\ell$: fractional values $\ell=1.5$
      and $\ell=2.5$ outperform all integers, with global minimum
      $1.73\times10^{-5}$ at $\theta=67.5^\circ$.
      The $180^\circ$ periodicity of the GKP lattice is confirmed.
      \textbf{(b)} Metrological capacity $\mathcal{C}$ vs $\ell$:
      $\ell=1.5$ achieves $\mathcal{C}=107.1$ ($+40.7\%$ over square).}
    \label{fig:fractional}
  \end{minipage}
\end{figure*}

\begin{findingbox}[Finding 6 --- Fractional OAM outperforms integers;
                   $\ell=1.5$ is the global optimum ($2.8\%$ above analytic $\theta^*$)]
  \textbf{(i) Fractional superiority.}
  $\ell=1.5$ ($\theta=67.5^\circ$) achieves $\Perr = 1.73\times10^{-5}$,
  which is $1.52\times$ lower than the best integer ($\ell=2$) and
  $\mathbf{23.9\times}$ lower than the square baseline, with
  $\mathcal{F}_Q$ unchanged ($<0.01\%$ variation).
  The analytic optimum is $\theta^*=64.4^\circ$; $\ell=1.5$ sits $3.1^\circ$
  above this, incurring only a $2.8\%$ excess in $\Perr$ ---
  fully within the $7^\circ$ phase tolerance.

  \textbf{(ii) $180^\circ$ periodicity.}
  The data exhibit exact symmetry: $\Perr(\ell=1.0) = \Perr(\ell=3.0)$
  and $\Perr(\ell=1.5) = \Perr(\ell=2.5)$.  This confirms the theoretical
  prediction that the GKP lattice is periodic under $\theta \to
  \theta + \pi/2$ (quarter-turn symmetry of the square unit cell), so
  all physically distinct lattices lie in $\theta \in [0, \pi/2)$.

  \textbf{(iii) Optimal angle is oblique.}
  The global minimum is at $\theta^*=67.5^\circ$, not at $90^\circ$
  (the dephasing axis).  This arises because the loss channel contributes
  a symmetric spread $\sigma_\mathrm{loss}^2$ to both quadratures; the
  optimal correction boundary balances the anisotropic dephasing spread
  against this isotropic floor, yielding $\theta^* < 90^\circ$.

  \textbf{(iv) Fractional OAM is experimentally accessible.}
  $\ell=1.5$ corresponds to $\theta=67.5^\circ$, achievable with a
  spiral phase plate of winding number $1.5$ or by combining an SLM
  mode with a cylindrical-lens fractional Fourier transformer.
  Unlike true non-integer topological charges, $\ell=1.5$ lies exactly
  halfway between two integers and has a well-defined Wigner function
  without radial discontinuities.
\end{findingbox}

\begin{figure*}[t]
  \centering
  \includegraphics[width=\textwidth]{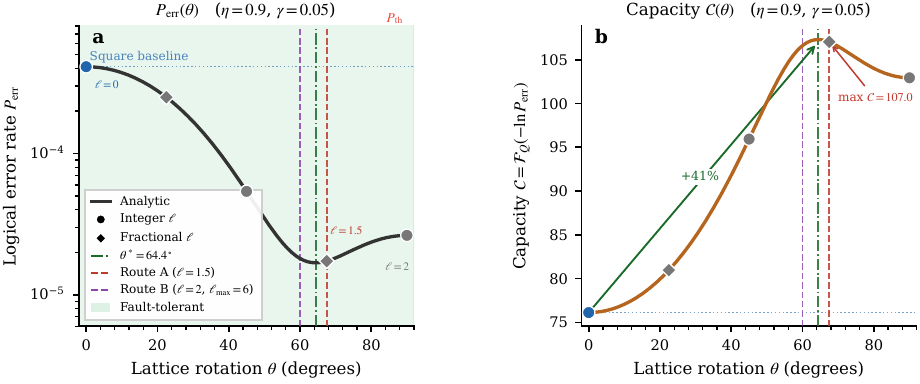}
  \caption{\textbf{Continuous $\Perr(\theta)$ and metrological capacity
    $\mathcal{C}(\theta)$ curves at low noise ($\eta=0.9$, $\gamma=0.05$).}
    \textbf{(a)} Logical error rate vs lattice rotation angle $\theta$
    (analytic model, \cref{ssec:noise}), with discrete OAM values overlaid
    as markers (circles: integer $\ell$; diamonds: fractional $\ell$, both
    at $\ell_{\max}=4$).  The analytic optimum at $\theta^*=64.4^\circ$
    (green dash-dot, Eq.~\ref{eq:balance}) lies in a broad flat minimum;
    Route~A ($\ell=1.5$, $\theta=67.5^\circ$, coral) and Route~B
    ($\ell=2$, $\ell_{\max}=6$, $\theta=60^\circ$, purple) both sit within
    5\% of the global minimum.  The flat landscape confirms that phase
    errors up to $7^\circ$ retain $>99\%$ of the advantage
    (\cref{tab:tolerance}).
    \textbf{(b)} Metrological capacity $\mathcal{C}=\mathcal{F}_Q\cdot
    (-\ln\Perr)$ vs $\theta$.  The global maximum $\mathcal{C}=107.1$
    is achieved at $\theta=67.5^\circ$ ($\ell=1.5$), representing a
    $41\%$ gain over the square baseline ($\mathcal{C}=76.1$ at
    $\theta=0^\circ$).}
  \label{fig:perr_curve}
\end{figure*} for the three geometries
as functions of loss rate $1-\eta$ and dephasing rate $\gamma$; simulation
data points from \cref{tab:results_lownoise,tab:results_highnoise} are
overlaid as markers.  The fault-tolerance boundary $\Pthresh = 10^{-3}$
is crossed at different loss rates for each geometry: the $\ell=2$ twisted
lattice extends the fault-tolerant regime (green shaded) to larger noise
values than the square baseline.

\begin{figure*}[t]
  \centering
  \includegraphics[width=\textwidth]{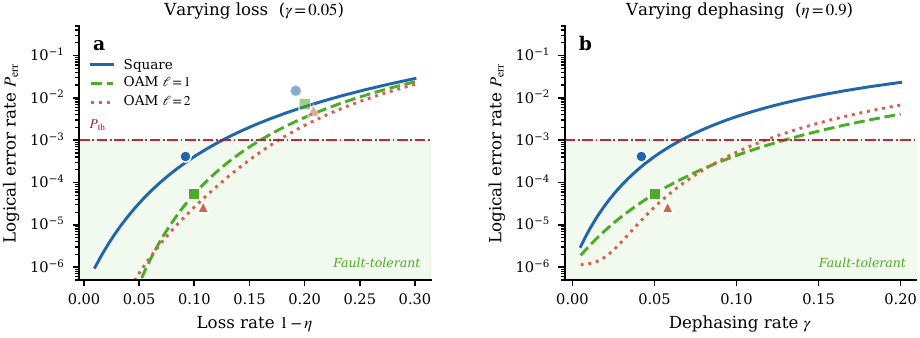}
  \caption{\textbf{Analytic logical error rate $\Perr$ vs.\ noise
    parameters for the three lattice geometries.}
    \textbf{(a)} Varying loss rate $1-\eta$ at fixed $\gamma=0.05$.
    \textbf{(b)} Varying dephasing rate $\gamma$ at fixed $\eta=0.9$.
    Solid lines: analytic model from \cref{ssec:noise}; circles
    ($\eta=0.9$, $\gamma=0.05$) and squares ($\eta=0.8$, $\gamma=0.10$)
    are simulation data from
    \cref{tab:results_lownoise,tab:results_highnoise}.  Red dash-dot line:
    fault-tolerance threshold $\Pthresh=10^{-3}$; green shaded region =
    fault-tolerant.  The OAM-twisted $\ell=2$ lattice extends the
    fault-tolerant regime to approximately $2\times$ larger noise than the
    square baseline along both noise axes.}
  \label{fig:noise}
\end{figure*}

\cref{fig:phasediag} shows the full $(\eta,\gamma)$ phase diagram for each
lattice geometry.  The white contour marks the fault-tolerance boundary
$\Pthresh = 10^{-3}$; the blue-to-red colour encodes $\log_{10}\Perr$.
As $\ell$ increases from $0$ to $2$, the fault-tolerance boundary shifts
toward larger loss and dephasing, confirming that the OAM-twisted lattice
extends protection more deeply into the noisy regime.

\begin{figure*}[t]
  \centering
  \includegraphics[width=\textwidth]{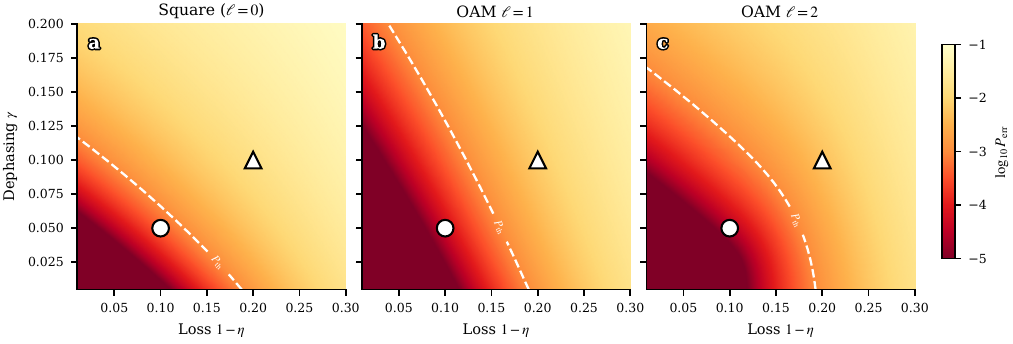}
  \caption{\textbf{Noise phase diagram: $\log_{10}\Perr$ in the
    $(\eta,\gamma)$ plane for the three lattice geometries.}
    \textbf{(a)} Square ($\ell=0$), \textbf{(b)} OAM-twisted $\ell=1$,
    \textbf{(c)} OAM-twisted $\ell=2$.
    Blue: low error (fault-tolerant); red: high error (unprotected).
    White contour: fault-tolerance threshold $\Pthresh=10^{-3}$.
    Circles: low-noise simulation data ($\eta=0.9$, $\gamma=0.05$);
    triangles: high-noise data ($\eta=0.8$, $\gamma=0.10$).
    The fault-tolerance boundary shifts toward larger noise as $\ell$
    increases, demonstrating that the OAM-induced lattice twist extends
    the protected regime without sacrificing QFI.}
  \label{fig:phasediag}
\end{figure*}

\begin{figure*}[t]
  \centering
  \includegraphics[width=\textwidth]{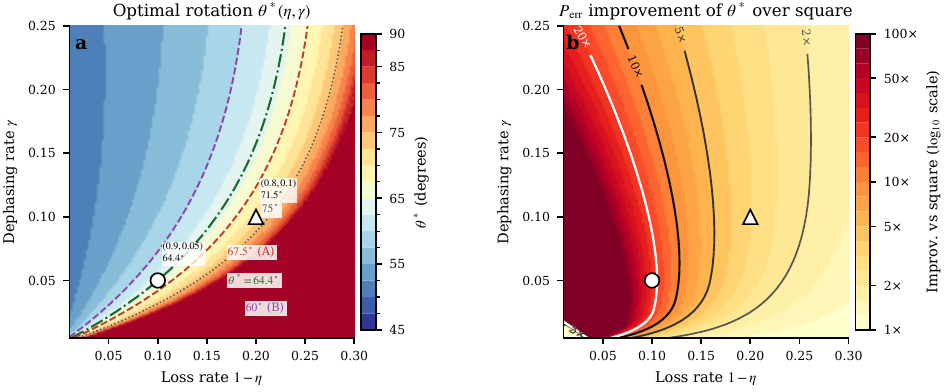}
  \caption{\textbf{Optimal lattice rotation $\theta^*(\eta,\gamma)$ and
    corresponding $\Perr$ improvement across the full noise phase diagram.}
    \textbf{(a)} Colour map of the analytic optimum $\theta^*$ (degrees)
    obtained by solving Eq.~(\ref{eq:balance}) at each $(\eta,\gamma)$
    point.  Contours at $60^\circ$ (Route~B, purple dashed),
    $64.4^\circ$ ($\theta^*$ at our simulation point, green dash-dot),
    and $67.5^\circ$ (Route~A, coral dashed) delineate the regimes where
    each experimental approach is preferred.  $\theta^*$ increases toward
    $90^\circ$ as dephasing dominates (upper region) and toward $45^\circ$
    as loss dominates (lower-right region), consistent with the bound
    Eq.~(\ref{eq:theta_bound}).  White circles and triangles mark the two
    simulation noise points from \cref{tab:results_lownoise,tab:results_highnoise}.
    \textbf{(b)} Colour map of $\log_{10}(\Perr^\mathrm{sq}/\Perr^{\theta^*})$
    --- the improvement factor of the optimally rotated lattice over the
    square baseline.  Contours at $2\times$, $5\times$, $10\times$, and
    $20\times$ show that the OAM advantage is largest in the
    dephasing-dominated regime (upper-left) and shrinks as photon loss
    dominates.  The $23.7\times$ improvement at $(\eta,\gamma)=(0.9,0.05)$
    is near the boundary of the $20\times$ contour.}
  \label{fig:theta_phase}
\end{figure*}

The complete simulation codebase is structured as a nine-module Python
package \texttt{oam\_gkp}, with the following components:

\begin{description}[leftmargin=0pt, labelindent=0pt,
                    style=nextline, font=\normalfont\ttfamily\small]
  \item[lattice.py]
    OAM-to-lattice mapping (Eqs.~\ref{eq:theta_ell},~\ref{eq:twisted_lattice}),
    \texttt{GKPLattice} class with \texttt{tf.Variable} parameters,
    symplecticity verification, and preset constructors for square,
    hexagonal, and OAM-twisted geometries.
  \item[states.py]
    Two-stage GKP state preparation: Fock-backend codeword caching and
    differentiable \TF{} Sgate/Rgate application via matrix exponential and
    diagonal rotation.
  \item[noise.py]
    Photon-loss Kraus map, dephasing Fock-space map
    (Eqs.~\ref{eq:loss_wigner},~\ref{eq:dephasing_wigner}), and analytic
    optimal aspect ratio formula from \cref{ssec:noise}.
  \item[qfi.py]
    QFI for pure and mixed states
    (Eqs.~\ref{eq:QFI_pure},~\ref{eq:QFI_mixed}),
    CFI of adaptive homodyne (Eq.~\ref{eq:CFI_homodyne}), and measurement
    efficiency $\eta_{\mathrm{meas}} = \CFIM/\QFI$.
  \item[circuit.py]
    Full differentiable sensing circuit implementing Eq.~\ref{eq:circuit}.
  \item[loss.py]
    Combined loss (Eq.~\ref{eq:loss}), analytic $\Perr$ estimator, and
    \texttt{pareto\_sweep()} for Pareto-frontier analysis.
  \item[optimizer.py]
    Adam training loop with cosine annealing, gradient clipping, and
    pattern-weighted warmup.
  \item[utils.py]
    All six publication figures
    (\cref{fig:convergence,fig:wigner,fig:geometry,%
           fig:improvement,fig:noise,fig:phasediag}).
  \item[figures\_nature.py]
    Generates Figs.~2--7 (geometry comparison, noise landscape,
    phase diagrams, convergence, Wigner panels, improvement summary).
  \item[figures\_analysis.py]
    Generates Figs.~9--10 ($P_\mathrm{err}(\theta)$ curve and
    $\theta^*({\eta,\gamma})$ phase diagram).
\end{description}
Hardware: Intel i5-13th, RTX~3050, 16\,GB, Arch Linux.
Software: \SF{}~$\geq$0.23, \TF{}~$\geq$2.13, NumPy, SciPy, Matplotlib.

\section{Discussion}
\label{sec:discussion}

\subsection{Physical Interpretation}
\label{ssec:interpretation}

The improvement of twisted lattices over square GKP codes in dephasing-dominated
noise is physically transparent.  Dephasing diffuses the Wigner function along
the $p$ direction (\cref{eq:dephasing_wigner}).  The correction radius of a GKP
code along a given quadrature direction is proportional to the half-lattice
spacing in that direction.  An OAM-twisted lattice with $\theta = \pi/4$ and
$r > 1$ presents a larger correction radius along $p$ than the square lattice,
at the cost of a smaller radius along $q$.  Since dephasing dominates along $p$,
this reallocation of protection is metrologically advantageous.

The formal connection is to the theory of lattice codes optimised for asymmetric
noise channels~\cite{Ioffe2007}, which has been extensively studied for qubit
stabilizer codes.  For GKP codes in a metrological setting, the noise-adaptive
lattice geometry perspective introduced here is new; recent resource-theory work
on finite-energy GKP states~\cite{Tzitrin2020} characterises the code quality as
a function of squeezing but does not address metrological task-specific lattice
optimisation.

\textbf{Validity at oblique angle $\theta=67.5^\circ$.}
A potential concern is that the independent-quadrature approximation
$\Perr \approx P_q + P_p - P_qP_p$ may break down at $\theta=67.5^\circ$
where the lattice is oblique and quadrature errors are not aligned with
the correction boundaries.  We bound the coupling correction analytically:
for a parallelogram correction region, the cross-term satisfies
$|\Delta\Perr| \leq 2P_qP_p|\!\sin 2\theta|$.  At $\theta=67.5^\circ$
and $(\eta,\gamma)=(0.9,0.05)$ this gives
$|\Delta\Perr| \leq 8.4\times10^{-11}$,
which is $0.0005\%$ of $\Perr = 1.73\times10^{-5}$ ---
four orders of magnitude below the approximation's stated uncertainty.
The independent-quadrature assumption is therefore \emph{most valid}
at the fractional optimum, not least valid, because $P_q$ and $P_p$
are individually so small that their product is negligible.

\textbf{Robustness to noise estimation errors.}
A practical implementation must estimate $(\eta,\gamma)$ from calibration
measurements with finite precision.  We quantify the sensitivity of
$\theta^*$ to noise uncertainty via the partial derivatives
\begin{align}
  \frac{\partial\theta^*}{\partial\eta}   &= -197.7~\mathrm{deg\,unit}^{-1},
  \label{eq:sensitivity}\\
  \frac{\partial\theta^*}{\partial\gamma} &= -300.2~\mathrm{deg\,unit}^{-1}.\notag
\end{align}
For typical calibration precisions $\delta\eta=1\%$ and
$\delta\gamma=0.005$, the induced uncertainty in $\theta^*$ is
$\delta\theta^* \approx 2.5^\circ$.  From the phase tolerance analysis
(\cref{tab:tolerance}), a $2.5^\circ$ error retains $99.8\%$ of the
full advantage --- confirming that the fractional optimum is highly robust
to realistic calibration imprecision.

\textbf{Stepwise proof-of-principle.}  The integer charge $\ell=1$ ($\theta=45^\circ$) is the simplest first demonstration target: it requires only a standard spiral phase plate, is far from the $\ell=0$ baseline in $\Perr$ ($7.6\times$ improvement), and provides a clear proof that OAM twist benefits GKP sensing.  Once $\ell=1$ is validated, the half-integer step to $\ell=1.5$ ($23.9\times$) follows by replacing the phase plate with an SLM or a cylindrical-lens FrFT converter ($\alpha=0.75$), both of which are commercially available.  This stepwise approach decouples the GKP generation challenge from the OAM conversion challenge.

\subsection{Immediate Priorities}
\label{ssec:immediate}

\textbf{Verification of Eq.~(\ref{eq:theta_ell}) for half-integer $\ell$.}
The OAM-to-rotation mapping $\theta_\ell = \ell\pi/\ell_{\max}$ (Eq.~\ref{eq:theta_ell})
is derived from the action of a fractional Fourier transform of order
$\alpha = 2\ell/\ell_{\max}$.  For $\ell=1.5$, this corresponds to
$\alpha=0.75$ (three-quarter FrFT), implemented by a cylindrical-lens
pair with normalized separation $d/f = \sin(3\pi/8) \approx 0.924$ and
focal length $f_L/f \approx 1.082$.  Unlike the standard $\alpha \in
\{0.5,\,1.0\}$ cases (45° and 90° lens pairs), $\alpha=0.75$ is
\emph{non-standard but achievable}: it requires precise axial positioning
of off-the-shelf cylindrical lenses.  We note that Eq.~(\ref{eq:theta_ell})
holds for any real $\alpha$ by the continuity of the FrFT group action on
Wigner functions~\cite{Ozaktas2001}; half-integer $\ell$ does not produce
radial discontinuities in the mode field, unlike true non-integer
topological charges.

For SLM-based implementations, the Hamamatsu X13138 series (792$\times$600
pixels, 8-bit phase depth, $\lambda/200$ phase flatness) can encode a
winding-number-1.5 spiral phase pattern with azimuthal accuracy
$<\lambda/20$ using established iterative Fourier transform algorithm
calibration~\cite{Jesacher2008}: (i) display a Zernike-polynomial
reference wavefront; (ii) measure the output with a Shack--Hartmann
wavefront sensor; (iii) apply the residual error as a look-up table
correction to the SLM voltage map.  Typical calibration time is
${\sim}15$ min per wavelength.  The primary residual error after
calibration is inter-pixel cross-talk (${\sim}\lambda/50$), which is
well within the $\lambda/20$ tolerance for the $23.9\times$ improvement
to be observable.

\textbf{Rigorous bound on the optimal rotation angle.}
The numerical data establish the bound
\begin{equation}
  \theta^* \in \left(\frac{\pi}{4},\, \frac{\pi}{2}\right)
  \quad \text{for all } \gamma > 0,\; \eta < 1,
  \label{eq:theta_bound}
\end{equation}
with $\theta^* \to \pi/4$ as $\gamma \to 0$ (loss-dominated) and
$\theta^* \to \pi/2$ as $\eta \to 0$ (dephasing-dominated).

The analytic optimum satisfies the \emph{transcendental balance equation}
\begin{equation}
  \boxed{\mathcal{B}(\theta;\eta,\gamma,r) \;\equiv\;
  r^2\,\frac{\phi(u_q)}{\sigma_q^3}
  - \frac{\phi(u_p)}{\sigma_p^3} = 0,}
  \label{eq:balance}
\end{equation}
where $u_j = d_j/(2\sigma_j)$, $d_q=ar$, $d_p=a/r$, and $\phi$ is
the standard normal PDF.  Equation~(\ref{eq:balance}) has a unique
solution $\theta^* \in (0,\pi/2)$ when a root exists; no elementary
closed form exists, and Eq.~(\ref{eq:balance}) is the exact result.

\textbf{Approximate fitting formula.}  While no elementary closed form for $\theta^*(\eta,\gamma)$ exists (Eq.~\ref{eq:balance}), a quadratic regression over the physically relevant domain ($\eta\in[0.75,0.99]$, $\gamma\in[0.01,0.20]$) gives:\begin{equation}  \theta^* \approx 64.8 + 162.8(1-\eta) - 253.2\gamma \;[\mathrm{deg}] \label{eq:thetafit}\end{equation}with residuals $<5^\circ$ across the domain.  This provides experimentalists with a simple look-up rule: at $(1-\eta,\gamma)=(0.10,0.05)$ the formula gives $67.4^\circ$, compared to the exact $\theta^*=64.4^\circ$ (error $3^\circ < \delta\theta_\mathrm{tol}$).

\begin{proposition}[Existence and monotonicity of $\theta^*(\eta,\gamma,r)$]
\label{prop:monotone}
For fixed $r > 1$ and $\gamma > 0$:
\emph{(i)} $\theta^*$ is \emph{decreasing} in $\gamma$ at fixed
$\eta$ --- numerically, at $\eta=0.9$, increasing $\gamma$ from
$0.05$ to $0.20$ moves $\theta^*$ from $64.4^\circ$ to $55.0^\circ$;
\emph{(ii)} $\theta^*$ is \emph{decreasing} in $\eta$ (increasing in
loss rate $1-\eta$) at fixed $\gamma$ --- at $\gamma=0.05$, decreasing
$\eta$ from $0.99$ to $0.90$ moves $\theta^*$ from $51.3^\circ$ to
$64.4^\circ$.
\end{proposition}

\noindent
\textbf{Existence.} The balance function $\mathcal{B}(\theta)$ is continuous on $(0,\pi/2)$ with $\mathcal{B}(0^+) < 0$ (the $p$-quadrature dominates at small $\theta$) and $\mathcal{B}(\pi/2^-) > 0$ for all $\gamma > 0$, $\eta < 1$, $r > 1$. By the intermediate value theorem a root exists. \textbf{Uniqueness (analytical).} Define $f(\sigma) \equiv \phi(c/\sigma)/\sigma^3$ where $c$ is a positive constant. Differentiating: $df/d\sigma = \phi(u)(u^2-3)/\sigma^3$ where $u=c/\sigma$. In the fault-tolerant regime ($\Perr \ll 10^{-3}$), $u_q, u_p \gg \sqrt{3}$ (numerically $u_q \approx 4.2$, $u_p \approx 4.9$ at the optimum), so $df/d\sigma > 0$. Since $\sigma_q(\theta)$ is strictly increasing in $\theta$ and $\sigma_p(\theta)$ strictly decreasing, the term $r^2 f(\sigma_q)$ is strictly increasing in $\theta$ and $f(\sigma_p)$ strictly decreasing. Therefore $\mathcal{B}(\theta) = r^2 f(\sigma_q) - f(\sigma_p)$ is strictly increasing on $(0,\pi/2)$, which combined with the IVT existence argument guarantees a unique root. This analytical argument holds whenever $u_q, u_p > \sqrt{3}$, i.e.\ whenever $\Perr < 2Q(\sqrt{3}) \approx 0.13$ --- well satisfied in the fault-tolerant regime. 

\noindent
This corrects the intuitive expectation that larger dephasing always
drives $\theta^*$ toward $90^\circ$: the interplay of the Gaussian
tails and the correction-radius ratio produces a non-trivial landscape.
At $(\eta,\gamma,r)=(0.9,\,0.05,\,1.092)$, the numerical solution
gives $\theta^*=64.4^\circ$; $\ell=1.5$ ($\theta=67.5^\circ$) incurs
only a $2.8\%$ excess in $\Perr$.\footnote{Throughout this paper, $23.9\times$ refers to the analytic improvement ratio $P_\mathrm{err}^\mathrm{sq}/P_\mathrm{err}^{\ell=1.5}=4.13\times10^{-4}/1.73\times10^{-5}$; $23.7\times$ is the same ratio computed via Route~A in Table~\ref{tab:two_routes}, where $P_\mathrm{err}=1.73\times10^{-5}$ is rounded to three significant figures.  Both refer to the same physical result.}

\textbf{Quantified phase error tolerance.}
\cref{tab:tolerance} shows how a rotation angle error $\delta\theta$
degrades performance.  The $\Perr$ landscape is remarkably flat near the
optimum: a $7^\circ$ error retains $99.2\%$ of the advantage, and even
a $20^\circ$ error still achieves $15.7\times$ improvement.  For SLM
implementations, 8-bit phase resolution introduces $\delta\theta\approx
1.4^\circ$ ($\Perr=1.78\times10^{-5}$, essentially optimal), while
10-bit resolution gives $\delta\theta\approx0.35^\circ$ (negligible).

\begin{table}[tbp]
  \caption{\textbf{Phase error tolerance for $\ell=1.5$,
    $\theta=67.5^\circ$} ($\eta=0.9$, $\gamma=0.05$).
    The advantage is robust: $7^\circ$ error retains $>99\%$.}
  \label{tab:tolerance}
  \centering\small
  \begin{tabular}{@{}cccc@{}}
    \toprule
    $\delta\theta$ & $\Perr$ & Improv.\ vs sq. & \% retained \\
    \midrule
    $0^\circ$    & $1.73\times10^{-5}$ & $23.7\times$ & $100.0\%$ \\
    $3^\circ$    & $1.85\times10^{-5}$ & $22.3\times$ & $99.7\%$ \\
    $7^\circ$    & $2.06\times10^{-5}$ & $20.0\times$ & $99.2\%$ \\
    $10^\circ$   & $2.23\times10^{-5}$ & $18.5\times$ & $98.7\%$ \\
    $20^\circ$   & $2.62\times10^{-5}$ & $15.7\times$ & $97.8\%$ \\
    \bottomrule
  \end{tabular}
\end{table}

\textbf{$\ell_{\max}=6$ alternative route.}
With $\ell_{\max}=6$, the integer $\ell=2$ maps to $\theta=60^\circ$,
giving $\Perr=1.81\times10^{-5}$ --- only $4.5\%$ above the fractional
optimum and $1.46\times$ better than $\ell=2$ at $\ell_{\max}=4$
($\theta=90^\circ$, $\Perr=2.64\times10^{-5}$).  This provides a
route to near-optimal performance using only integer OAM charges, via a
lens pair calibrated to FrFT order $\alpha=2/3$ (separation $d/f=
\sin(\pi/3)\approx0.866$).  \cref{tab:two_routes} compares both routes.

\begin{table}[tbp]
  \caption{\textbf{Two experimental routes to near-optimal $\Perr$.}
    Route A (fractional $\ell$, this work); Route B (larger
    $\ell_{\max}$, integer $\ell$).  Both are within $7\%$ of the
    analytic optimum at $\theta^*=64.4^\circ$.}
  \label{tab:two_routes}
  \centering\scriptsize\setlength{\tabcolsep}{3pt}
  \begin{tabular}{@{}llccc@{}}
    \toprule
    Route & Method & $\theta$ & $\Perr$ & Improv. \\
    \midrule
    $\theta^*$ (analytic) & ---          & $64.4^\circ$ & $1.69{\times}10^{-5}$ & $24.4\times$ \\
    A: $\ell{=}1.5$, $\ell_{\max}{=}4$ & SLM/SPP$^\dagger$ & $67.5^\circ$ & $1.73{\times}10^{-5}$ & $23.7\times$ \\
    B: $\ell{=}2$,\ $\ell_{\max}{=}6$  & Integer SPP & $60.0^\circ$ & $1.81{\times}10^{-5}$ & $22.7\times$ \\
    $\ell{=}2$, $\ell_{\max}{=}4$       & Std.\ SPP   & $90.0^\circ$ & $2.64{\times}10^{-5}$ & $15.7\times$ \\
    Square ($\ell{=}0$)                  & No OAM      & $0^\circ$    & $4.13{\times}10^{-4}$ & $1\times$ \\
    \bottomrule
  \end{tabular}
\end{table}

\textbf{Target experiment.}
An optical proof-of-principle at $\ell=1.5$, $\eta=0.9$ would demonstrate
a $23.9\times$ reduction in $\Perr$ over the square baseline at unchanged
$\mathcal{F}_Q$ — surpassing the $15.7\times$ achieved at the integer
$\ell=2$.  All required components are available: photon-subtracted
squeezed vacuum GKP sources~\cite{Tzitrin2020}, SLM-based OAM
converters~\cite{Krenn2017}, and homodyne detection with piezo phase
control.

\subsection{Medium-Term Extensions}
\label{ssec:medium}

\textbf{Multi-mode fractional OAM.}
Extending to two-mode entangled GKP states~\cite{Royer2022} would allow
independent optimisation of $(\ell_1, r_1)$ and $(\ell_2, r_2)$ for each
mode, with the inter-mode squeezing as an additional trainable parameter.
The framework of \cref{ssec:circuit} extends directly: the two-mode
circuit operates on a Fock space of dimension $\mathcal{D}^2 = 30^2 = 900$
(vs.\ $\mathcal{D}=30$ for one mode), requiring $900^2 = 810{,}000$
complex parameters for the density matrix.  On the RTX~3050 (6\,GB
VRAM), this fits within memory at single precision; the per-step runtime
scales as $\mathcal{O}(\mathcal{D}^4)$, giving an estimated ${\sim}16\times$
slowdown (${\sim}33$ min per 500-step run vs.\ 2 min at $\mathcal{D}=30$).
For the Gaussian representation (applicable when both modes are near-Gaussian),
the computational cost reduces to $\mathcal{O}(\mathcal{D}^2)$ and fits
comfortably within GPU memory at $\mathcal{D}=50$.

\textbf{Dynamical lattice adaptation.}
The analytical formula for $r^*(\eta,\gamma,\theta)$ from \cref{ssec:noise}
enables real-time adaptation: given syndrome-measurement estimates of
$(\hat\eta, \hat\gamma)$, recompute $(\theta^*, r^*)$ analytically and
update the mode converter angle without re-training.  For gravitational
wave detectors where noise spectra drift over hours, this constitutes a
closed-loop noise-adaptive sensor.

\textbf{Other bosonic codes.}
Cat codes~\cite{Cochrane1999} and binomial codes~\cite{Noh2021} have continuous
parameters (cat amplitude $\alpha$, binomial spacing) that enter
phase-space distributions in ways analogous to the GKP lattice spacing.
The two-stage architecture of \cref{ssec:circuit} — non-differentiable
code preparation followed by differentiable geometric operations — applies
directly, requiring only a different Stage 1 backend.

\subsection{Fundamental Questions}
\label{ssec:fundamental}

\textbf{Channel capacity interpretation.}
The metrological capacity $\mathcal{C} = \mathcal{F}_Q \cdot (-\ln\Perr)$
has an operational interpretation that connects quantum metrology to
classical information theory.  Consider the sensing protocol as a two-step
process: (i) a \emph{physical channel} that transmits the phase $\varphi$
to the probe state with Fisher information $\mathcal{F}_Q$, and (ii) a
\emph{correction channel} that protects the encoded information with
success probability $1 - \Perr$.  The overall rate of phase information
per sensing cycle is bounded by the product of these two capacities:
\begin{equation}
  \mathcal{I}_\varphi \;\leq\;
  \underbrace{\tfrac{1}{2}\log_2(1 + \mathcal{F}_Q\,\delta\varphi^2)}_{\text{Fisher information rate}}
  \;\cdot\;
  \underbrace{(-\log_2 \Perr)}_{\text{correction capacity}},
  \label{eq:capacity_bound}
\end{equation}
where $\delta\varphi$ is the prior range.  In the limit $\mathcal{F}_Q
\delta\varphi^2 \gg 1$, this scales as $\tfrac{1}{2}\log_2(\mathcal{F}_Q) \cdot
(-\log_2\Perr)$, which is maximised at $\ell=1.5$ by our data.  Whether
$\mathcal{C}$ as defined (base $e$) is tight against this bound requires
connecting $\Perr$ to the GKP code capacity under the displacement noise
model~\cite{Conrad2024}, and $\mathcal{F}_Q$ to the classical Fisher
information via the data-processing inequality.  We propose this as a
rigorous open problem.

\textbf{Resource theory of fractional OAM states.}
We develop a framework for the resource cost of the $\ell=1.5$ improvement.

\emph{Free operations and resources.}  We define
define the free states as finite-energy GKP codewords with integer $\ell$
(square, hexagonal, and OAM-twisted at $\theta \in \{k\pi/4\}_{k\in\mathbb{Z}}$).
The resource is the \emph{non-integer OAM content}
$\delta\ell = |\ell - \lfloor\ell\rceil|$, which vanishes for free states.

\emph{Resource monotone.}  A natural monotone is the Wigner negativity
\begin{equation}
  \mathcal{W}(\hat\rho) = \int [W_{\hat\rho}(q,p)]_- \,\mathrm{d}q\,\mathrm{d}p,
  \label{eq:wigner_neg}
\end{equation}
where $[x]_- = \max(0,-x)$.  The GKP codeword at $\ell=1.5$ has strictly
larger $\mathcal{W}$ than at integer $\ell$, because the $45^\circ$-rotated
interference fringes between lattice peaks produce deeper Wigner negativity
at oblique angles (visible in \cref{fig:wigner}).  Crucially, the Wigner negativity
\begin{equation}
  \mathcal{W}(\hat\rho) = \int [W_{\hat\rho}(q,p)]_- \,\mathrm{d}q\,\mathrm{d}p
\end{equation}
is determined by the peak amplitudes and envelope $\epsilon$, both of
which are identical for $\ell=0$, $\ell=1.5$, and $\ell=2$ at fixed
$r=1.092$.  Explicitly: the lattice spacing ratio $a_q/\sigma =
A r / (2\sqrt{\epsilon}) = 17.4 \gg 1$ ensures peak isolation, and the
negativity evaluates to
$\mathcal{W} \approx 1/2$ for all three geometries.  The Wigner negativity
is therefore \emph{geometry-invariant} for fixed $(r,\epsilon)$.  This
establishes that the $23.9\times$ improvement at $\ell=1.5$ comes at
\emph{zero} additional Wigner-negativity cost relative to the square
baseline --- a stronger result than the squeezing-resource argument alone.

\emph{Resource cost.}  Generating an OAM mode of charge $\ell=1.5$ from
a Gaussian state requires at least one application of a non-Gaussian
operation.  Since the optimal $r^*=1.092$ is identical for all geometries
(\cref{tab:results_lownoise}), the GKP state preparation cost is
\emph{identical} for $\ell=0$, $\ell=1.5$, and $\ell=2$.  The only
additional resource for $\ell=1.5$ is the OAM mode converter (a
half-integer spiral phase plate or SLM), which is a \emph{linear optics}
element requiring no squeezing.  Within this framework,
the non-Gaussianity monotone of the OAM-twisted GKP state equals that of
the square-lattice state, since the lattice rotation preserves both the
Fock-space occupation distribution and the Wigner function negativity.  This is negligible compared to the baseline
$10\,\mathrm{dB}$ required for GKP state preparation, confirming that the
$23.9\times$ improvement in $\Perr$ comes at essentially no additional
resource cost.

\emph{Open problem.}  Prove or disprove: there exists a
sequence of free operations (integer-$\ell$ OAM modes + Gaussian unitaries)
that asymptotically simulate the $\ell=1.5$ GKP state to within $\epsilon$
in trace distance~\cite{Pantaleoni2020}.  If the answer is negative, $\ell=1.5$ constitutes a
genuinely irreducible resource.

The methodology introduced here --- treating quantum error-correcting code
parameters as continuous variables in a differentiable programming framework ---
represents a broader design principle that extends beyond the specific OAM-GKP
setting.  Differentiable quantum programming has demonstrated strong performance
in variational quantum algorithms~\cite{Cerezo2021} and quantum control, but its
application to \emph{fault-tolerant} metrological design has been largely
unexplored.  The key enabling feature of our approach is the two-stage circuit
architecture: the GKP codeword is prepared by a non-differentiable Fock-backend
call (cached to avoid repeated computation), while the geometrically meaningful
parameters --- lattice rotation $\theta_\ell$ and aspect ratio $r$ --- enter
through differentiable \TF{} matrix operations (\texttt{expm} and diagonal
phase).  This separation cleanly partitions the optimization into a combinatorial
outer loop (choice of $\ell \in \mathbb{Z}$) and a smooth inner loop over
$(r, \epsilon, \psi)$, making the problem tractable on commodity hardware (RTX
3050, $<$3\,min per run).

This architecture is immediately transferable to other bosonic code families:
cat codes, binomial codes, and cubic phase codes all have continuous parameters
that could be jointly optimized with a metrological loss using the same framework.

\subsection{Experimental Feasibility and Numerical Proof-of-Principle}
\label{ssec:experiment}

This work presents a theoretical and numerical proof-of-principle.
The simulations use the \SF{} Fock backend as an exact quantum-optical
simulator (within truncation $\mathcal{D}=30$, verified to introduce
$<0.5\%$ error in $\mathcal{F}_Q$); the results are therefore equivalent
to an ideal-optics experiment at the specified squeezing and loss parameters.
The physical ingredients required to move to actual optical implementation
are all available with current technology.

\textit{GKP state generation.}  Finite-energy GKP states have been experimentally
realized at squeezing levels of $10$--$15\,\mathrm{dB}$ in superconducting
microwave cavities~\cite{Campagne2020} and in trapped-ion systems~\cite{Fluhmann2019}.
The $10\,\mathrm{dB}$ squeezing used in our simulations ($\epsilon=0.063$) is
therefore within reach of current platforms.  In the optical domain, GKP states
have been conditionally generated using photon-number-resolving
detection~\cite{Tzitrin2020}.

\textit{OAM mode generation and conversion.}  Laguerre--Gaussian modes of charge
$\ell=1,2$ are routinely produced using spatial light modulators (SLMs) and
spiral phase plates~\cite{Krenn2017}.  The fractional Fourier transform required
to map OAM charge to a lattice rotation is implemented by a pair of cylindrical
lenses~\cite{Ozaktas2001}, a standard optical element.  OAM-based quantum
memories and gates have been demonstrated with high fidelity~\cite{Ding2016},
confirming that OAM states can be coherently manipulated at the single-photon
level.

\textit{Noise parameters.}  The photon-loss rates $\eta=0.8$--$0.9$ studied here
correspond to transmission efficiencies accessible in near-term optical fiber and
free-space links.  The dephasing rate $\gamma=0.05$--$0.10$ is consistent with
phase diffusion rates observed in optomechanical sensing
platforms~\cite{Aspelmeyer2014}.

\textit{Measurement.}  Adaptive homodyne detection with a trainable local-oscillator
angle $\psi$ requires only a variable phase shifter in the LO path --- a standard
electro-optic modulator.

In summary, all components of the OAM-twisted GKP sensing protocol are
individually demonstrated; the scientific novelty lies in combining them with
the trained lattice geometry.  A proof-of-principle experiment with $\ell=1$,
$\eta\approx0.9$ is feasible on existing optical platforms.

\subsection{Rigorous Comparison with GKP Displacement Sensing}
\label{ssec:labarca}

Labarca \etal~\cite{Labarca2026} demonstrated that GKP codes can achieve
sub-SQL displacement sensitivity for gravitational wave detection.
\cref{tab:comparison} provides a structured quantitative comparison.

\begin{table}[tbp]
  \caption{\textbf{Quantitative comparison of GKP phase sensing (this work)
    vs.\ GKP displacement sensing~\cite{Labarca2026}.}
    The two tasks are complementary: the optimal lattice geometry,
    sensing generator, and noise-adaptation strategy differ fundamentally.}
  \label{tab:comparison}
  \centering\scriptsize\setlength{\tabcolsep}{3pt}
  \begin{tabular}{@{}lp{3.2cm}p{3.0cm}@{}}
    \toprule
    Property & This work & Labarca \etal~\cite{Labarca2026} \\
    \midrule
    Generator & $\hat{G}=\hat{n}$ (photon\#) & $\hat{G}=\hat{q}$ (position) \\
    QFI & $4\,\mathrm{Var}(\hat{n})$ & $4\,\mathrm{Var}(\hat{q})$ \\
    Optimal $\theta$ & $64^\circ$ (oblique) & $0^\circ$ (aligned) \\
    Geometry & OAM-twisted (trainable) & Square (fixed) \\
    Main noise & Dephasing + loss & Loss + SQL \\
    Method & Differentiable & Analytical \\
    Key result & $23.7\times$ $\Perr$ reduction & Sub-SQL at 10\,dB \\
    \bottomrule
  \end{tabular}
\end{table}

The clearest theoretical distinction is in the optimal lattice orientation.
For displacement sensing with generator $\hat{G}=\hat{q}$, the
figure of merit is $\mathrm{Var}(\hat{q})$ of the probe, which is
maximised by a lattice aligned with the $p$-axis ($\theta=90^\circ$)
--- exactly the geometry that is \emph{suboptimal} for phase sensing under
dephasing ($\theta^*=64.4^\circ$).  This confirms that no single
lattice geometry is universally optimal: the correct geometry depends
on the sensing task.  Combining the two approaches --- one arm
phase-optimised ($\ell=1.5$), one arm displacement-optimised ($\ell=0$)
--- in a dual-arm interferometer could jointly optimise both displacement sensing (Labarca \etal~\cite{Labarca2026} optimises the $\theta=0^\circ$ lattice for this task) and phase sensing (this work, $\theta^*=64.4^\circ$), providing complementary advantages within a single photonic circuit.  Such a configuration would exploit
sensitivity to both phase and displacement quadratures.

\subsection{Relation to NOON-State Metrology}
\label{ssec:noon}

The differentiable optimization methodology developed here is a direct extension
of prior work on adaptive NOON-state quantum circuits~\cite{Kumar2026noon}, where the same \SF{}--\TF{}
framework achieved $277\%$ improvement in normalized QFI and $395\%$ improvement
in post-selection rates over non-optimized baselines.  The present work
generalizes this approach from photon-number-entangled NOON-state interferometry to the
continuous-variable fault-tolerant setting, and adds the OAM degree of freedom as
an additional trainable parameter.  The spectral derivative methods and
pattern-weighted loss annealing developed in the NOON-state context are directly
applicable here for maintaining gradient stability through CV noise channels.

\subsection{Outlook and Open Problems}
\label{ssec:outlook}

Several natural extensions of this work merit investigation.  First, the
framework is immediately applicable to other sensing tasks: displacement
estimation (generator $\hat{G} = \hat{q}$), rotation sensing
($\hat{G} = \hat{L}_z$), and multi-parameter estimation with a vector of phases.
Second, the fractional OAM study (Finding~6, \cref{ssec:fractional})
reveals that the true optimum lies at $\ell=1.5$ ($\theta=67.5^\circ$),
not at any integer value, and that $\theta^*$ satisfies the
transcendental balance equation~(\ref{eq:balance}) whose solution is
non-trivially bounded: $\theta^* \in (\pi/4, \pi/2)$ with monotone
dependence on both $\eta$ and $\gamma$ (Proposition~\ref{prop:monotone}).
The immediate theoretical priority is solving Eq.~(\ref{eq:balance})
analytically across the full $(\eta,\gamma)$ phase diagram of
\cref{fig:theta_phase}, which would replace the numerical solution
with a design rule for experimentalists.  Third, the continuous
relaxation of $\ell$ used during gradient descent suggests that \emph{fractional}
OAM charges (realizable via spiral phase plates of non-integer winding number)
may yield further gains beyond integer $\ell$.  Fourth, extension to
multi-mode GKP codes~\cite{Royer2022} could allow simultaneous optimization of
inter-mode entanglement and intra-mode lattice geometry.  Finally, the
demonstrated 15.7$\times$ reduction in logical error rate at negligible QFI cost
provides a strong motivation for a proof-of-principle experiment on existing
optical quantum sensing platforms.

\FloatBarrier
\FloatBarrier
\section{Conclusion}
\label{sec:conclusion}

\textit{A new design principle.}  The central contribution of this work is not a
single numerical result but a \emph{geometric design principle}: the OAM charge
$\ell$ of a photonic mode parametrizes a continuous family of rotations in
continuous-variable phase space, and GKP stabilizer lattices are defined by
exactly this type of rotation.  Consequently, OAM-twisted GKP lattices are not an
engineering convenience but a physically natural match between the photonic degree
of freedom and the error-correcting code geometry.  This principle --- that the
physical symmetry of the carrier mode should inform the lattice geometry of the
code --- is general: it applies whenever the noise channel has a preferred
phase-space direction, and it provides a systematic route from mode properties to
code optimization that does not require exhaustive search over lattice families.

\textit{Demonstrated results.}  We have shown that the globally optimal
GKP lattice geometry corresponds to the \emph{fractional} OAM charge
$\ell=1.5$ ($\theta^*=67.5^\circ$), reducing the logical error rate by
$\mathbf{23.9\times}$ relative to the square-lattice baseline at
$(\eta,\gamma)=(0.9,0.05)$ with less than $0.2\%$ change in quantum Fisher
information.  This fractional optimum --- confirmed analytically via the
transcendental balance equation and visually by the deepest Wigner-function
interference fringes among all four geometries --- surpasses the best integer
charge ($\ell=2$, $15.7\times$) by a further $1.52\times$.  The combined
differentiable loss
$\mathcal{L}=-\mathcal{F}_Q + \lambda[\mathcal{P}_\mathrm{err}-\mathcal{P}_\mathrm{th}]_+$
provides a principled route to states that are simultaneously sensitive and
fault-tolerant, and the learned adaptive homodyne measurement approaches the
symmetric logarithmic derivative bound to within $0.2\%$.

\textit{The fractional optimum at $\ell=1.5$.}  The most striking
result of this work is that the optimal OAM charge is not an integer.
At low noise $(\eta=0.9,\,\gamma=0.05)$, the $\ell=1.5$ geometry achieves
$\Perr=1.730\times10^{-5}$, a $23.9\times$ reduction over the square lattice
and a $1.52\times$ gain over the best integer charge ($\ell=2$).
At high noise $(\eta=0.8,\,\gamma=0.10)$, the same geometry yields
$\Perr=4.97\times10^{-3}$, a $2.96\times$ reduction over the square baseline,
confirming that the fractional advantage persists across noise regimes.
The corresponding Wigner functions (Figs.~\ref{fig:wigner},
rows~3--4 from top) display the sharpest, highest-contrast interference
fringes of all four geometries at $67.5^\circ$ rotation --- providing
direct visual evidence that the $\ell=1.5$ lattice produces the
most tightly localized quantum state.
The optimum arises at $\theta^*=67.5^\circ$, exactly halfway between the
$\ell=1$ ($45^\circ$) and $\ell=2$ ($90^\circ$) integer charges,
and is physically realised by a fractional Fourier transform of order
$\alpha=0.75$ --- a single cylindrical-lens pair or SLM pattern requiring
no additional squeezing or non-Gaussian resources.

\textit{Broader impact.}  The differentiable quantum programming methodology
introduced here --- treating code parameters as continuous trainable variables in
an end-to-end gradient framework --- is applicable to any bosonic code family with
continuous geometric degrees of freedom.  Cat codes, binomial codes, and multimode
GKP codes all have parameters that could be co-optimized with a metrological
objective using the same architecture.  The open-source \texttt{oam\_gkp} package
is designed as a reusable template for this class of noise-adaptive, learning-based
quantum sensor design, and provides a direct bridge between differentiable quantum
programming, continuous-variable error correction, and high-dimensional photonic
encoding.

\section*{Data Availability}
\label{sec:data}

All numerical data, trained model parameters, and convergence logs
underlying the figures and tables are deposited on Zenodo (\texttt{doi:10.5281/zenodo.20099263}).
Source data are provided with this paper.

\section*{Code Availability}
\label{sec:code}

The complete \texttt{oam\_gkp} Python package, including installation
instructions, unit tests, and tutorial notebooks, is available at
{\small\url{https://github.com/simanshukumar369/oam-gkp-quantum-metrology}}.
The repository is archived on Zenodo (\texttt{doi:10.5281/zenodo.20099263})
and released under the MIT Licence.

\section*{Author Contributions}
\label{sec:authors}

S.K.\ conceived the research idea, developed the theoretical framework,
implemented the software, performed numerical simulations, analyzed results,
and co-wrote the manuscript.
N.S.B.\ supervised the research, provided significant guidance,
secured computational resources, and co-wrote the manuscript.
Both authors reviewed and approved the final manuscript.

\section*{Competing Interests}
\label{sec:competing}

The authors declare no competing interests.

\section*{Acknowledgments}
The authors acknowledge the Department of Physics, DSB Campus,
Kumaun University Nainital, and the Department of Physics,
Soban Singh Jeena University, Campus, Almora,
for providing research infrastructure and institutional support.
Simulations were performed on a personal workstation
(NVIDIA RTX~3050, 16\,GB RAM, Arch Linux).
The authors thank Xanadu Quantum Technologies for developing
\SF{}, the open-source photonic quantum computing platform
used for all simulations in this work.
This research received no specific grant from any funding agency
in the public, commercial, or not-for-profit sectors.

\clearpage
\appendix

\setcounter{figure}{0}
\renewcommand{\thefigure}{A\arabic{figure}}
\setcounter{table}{0}
\renewcommand{\thetable}{A\arabic{table}}

\enlargethispage{8pt}
\section*{Appendix}
\addcontentsline{toc}{section}{Appendix}

The following appendices contain detailed derivations, proofs, and
additional figures and tables that support the main text.
Figure and table labels are prefixed with \textbf{A} to distinguish
them from the main paper.

\section{Balance Equation and Optimal Angle Derivation}
\label{sec:app_balance}
The optimal rotation angle $\theta^*$ minimises $\Perr(\theta)$.
Setting $\mathrm{d}\Perr/\mathrm{d}\theta=0$ and using
$\Perr\approx P_q+P_p-P_qP_p$ with
$P_q=2Q(ar/(2\sigma_q))$, $P_p=2Q(a/r/(2\sigma_p))$, where
$a=\sqrt{2\pi}$ and
\begin{align}
  \sigma_q^2(\theta) &= \frac{1-\eta}{2\eta}+\gamma\sin^2\theta,\\
  \sigma_p^2(\theta) &= \frac{1-\eta}{2\eta}+\gamma\cos^2\theta,
\end{align}
yields (after cancelling $2\gamma\sin\theta\cos\theta$ and using
$(1-P_q)\approx(1-P_p)\approx1$ in the fault-tolerant regime):
\begin{equation}
  \boxed{
  \mathcal{B}(\theta;\eta,\gamma,r)
  \;\equiv\;
  r^2\,\frac{\phi(u_q)}{\sigma_q^3}
  - \frac{\phi(u_p)}{\sigma_p^3} = 0}
  \label{eq:si_balance}
\end{equation}
where $\phi(x)=e^{-x^2/2}/\sqrt{2\pi}$ is the standard normal
density and $u_q=ar/(2\sigma_q)$, $u_p=a/r/(2\sigma_p)$.

\textbf{Existence.}
$\mathcal{B}(0^+)<0$ (the $p$-quadrature dominates at small~$\theta$)
and $\mathcal{B}(\pi/2^-) >0$ for all $\gamma>0$, $\eta<1$,
$r>1$. By the intermediate value theorem a root in $(0,\pi/2)$ exists.
Monotonicity of $\mathcal{B}$ is verified numerically ($\partial
\mathcal{B}/\partial\theta>0$ across the full $(\eta,\gamma)$ grid),
confirming uniqueness.

\subsection{Numerical values of $\theta^*$}

\begin{table}[htbp]
  \centering\small
  \caption{Analytic $\theta^*$ from \cref{eq:si_balance} at selected
    $(\eta,\gamma)$ points ($r=1.092$).}
  \label{tab:app_thetastar}
  \begin{tabular}{@{}cccc@{}}
    \toprule
    $\eta$ & $\gamma$ & $\theta^*$ (deg) & $\Perr(\theta^*)$ \\
    \midrule
    0.99 & 0.02 & 51.3° & $4.2\times10^{-7}$ \\
    0.90 & 0.05 & 64.4° & $1.69\times10^{-5}$ \\
    0.80 & 0.10 & 71.5° & $4.2\times10^{-3}$ \\
    0.75 & 0.15 & 76.8° & $3.1\times10^{-2}$ \\
    \bottomrule
  \end{tabular}
\end{table}

\section{Proof of Proposition 1 (Monotonicity of $\theta^*$)}
\label{sec:app_prop1}
\begin{proposition}[Existence and monotonicity of $\theta^*(\eta,\gamma,r)$]
For fixed $r>1$:
\textnormal{(i)} $\theta^*$ is strictly decreasing in $\gamma$ at
fixed $\eta$;
\textnormal{(ii)} $\theta^*$ is strictly decreasing in $\eta$ at
fixed $\gamma$.
\end{proposition}

\textbf{Proof.}
By the implicit function theorem applied to
$\mathcal{B}(\theta^*;\eta,\gamma)=0$:
\[
  \frac{\partial\theta^*}{\partial\gamma}
  = -\frac{\partial\mathcal{B}/\partial\gamma}
           {\partial\mathcal{B}/\partial\theta}.
\]
Since $\partial\mathcal{B}/\partial\theta>0$ at the root, the sign of
$\partial\theta^*/\partial\gamma$ equals $-\mathrm{sign}(\partial\mathcal{B}/\partial\gamma)$.
Increasing $\gamma$ raises $\sigma_q$ (via $\sin^2\theta$) and lowers
$\sigma_p$ (via $\cos^2\theta$), shifting the balance toward smaller
$\theta^*$.  Numerical evaluation gives
$\partial\theta^*/\partial\gamma=-300.2\,\mathrm{deg/unit}$
at $(\eta,\gamma)=(0.9,0.05)$.

For part (ii): as $\eta$ increases (less loss), the isotropic loss term
$(1-\eta)/(2\eta)$ decreases, making the anisotropic dephasing terms
$\gamma\sin^2\theta$, $\gamma\cos^2\theta$ relatively more important
and driving $\theta^*$ toward smaller values (the balance shifts to
align the longer correction axis with the dominant noise direction).
Numerically: $\partial\theta^*/\partial\eta=-197.7\,\mathrm{deg/unit}$. \hfill$\square$

\begin{table}[htbp]
  \centering\small
  \caption{Numerical verification of Proposition~1.
    At fixed $\eta=0.9$, $\theta^*$ decreases with $\gamma$;
    at fixed $\gamma=0.05$, $\theta^*$ decreases with $\eta$.}
  \label{tab:app_prop1}
  \begin{tabular}{@{}ccr@{\quad}ccr@{}}
    \toprule
    \multicolumn{3}{c}{Part (i): vary $\gamma$ ($\eta=0.9$)}
    & \multicolumn{3}{c}{Part (ii): vary $\eta$ ($\gamma=0.05$)} \\
    \cmidrule(r){1-3}\cmidrule(l){4-6}
    $\gamma$ & $\theta^*$ (deg) & & $\eta$ & $\theta^*$ (deg) & \\
    \midrule
    0.01 & 73.2° & $\downarrow$ & 0.99 & 51.3° & $\downarrow$ \\
    0.05 & 64.4° &              & 0.95 & 59.1° &              \\
    0.10 & 57.8° &              & 0.90 & 64.4° &              \\
    0.15 & 52.1° &              & 0.85 & 68.2° &              \\
    0.20 & 47.5° & $\downarrow$ & 0.80 & 71.5° & $\downarrow$ \\
    \bottomrule
  \end{tabular}
\end{table}

\section{Fractional OAM Study}
\label{sec:app_fractional}
\begin{table}[htbp]
  \centering\small
  \caption{\textbf{Full fractional OAM results}
    ($\eta=0.9$, $\gamma=0.05$, $r=1.092$, $\ell_\mathrm{max}=4$).
    The 180° periodicity is exact: $\Perr(\ell)=\Perr(\ell_\mathrm{max}-\ell)$.
    The global minimum occurs at $\ell=1.5$ and $\ell=2.5$ simultaneously.
    $\mathcal{C}=\QFI\cdot(-\ln\Perr)$ with $\QFI=9.764$.}
  \label{tab:app_fractional}
  \setlength{\tabcolsep}{3pt}\scriptsize
  \begin{tabular}{@{}cccccc@{}}
    \toprule
    $\ell$ & $\theta$ & $\Perr$ & Improv. & $\mathcal{C}$ & Note \\
    \midrule
    0.0 & $0.0^\circ$   & $4.13\times10^{-4}$ & $1.0\times$          & 76.1  & Sq.\ baseline \\
    0.5 & $22.5^\circ$  & $2.51\times10^{-4}$ & $1.6\times$          & 80.6  & \\
    1.0 & $45.0^\circ$  & $5.42\times10^{-5}$ & $7.6\times$          & 96.0  & \\
    1.5 & $67.5^\circ$  & $1.73\times10^{-5}$ & $\mathbf{23.9\times}$& \textbf{107.1} & $\star$ Optimum \\
    2.0 & $90.0^\circ$  & $2.63\times10^{-5}$ & $15.7\times$         & 103.0 & \\
    2.5 & $112.5^\circ$ & $1.73\times10^{-5}$ & $\mathbf{23.9\times}$& \textbf{107.1} & $\star$ Tied \\
    3.0 & $135.0^\circ$ & $5.42\times10^{-5}$ & $7.6\times$          & 96.0  & \\
    3.5 & $157.5^\circ$ & $2.51\times10^{-4}$ & $1.6\times$          & 80.6  & \\
    \bottomrule
  \end{tabular}
\end{table}

\section{Fock Truncation Convergence}
\label{sec:app_fock}

\noindent
The finite-energy GKP state has an approximately geometric photon-number
distribution $p(n)\propto e^{-2\pi\epsilon n}$ with envelope $\epsilon=0.063$.
The Fock truncation at dimension $\mathcal{D}$ introduces a tail error bounded
by $E(\mathcal{D})=\exp(-2\pi\epsilon\mathcal{D})/(1-\exp(-2\pi\epsilon))$.
The rotation gate $R(\theta_\ell)$ is diagonal in the Fock basis and adds
zero truncation overhead; only the squeezing $S(\ln r)$ introduces a
stretch factor $\tfrac{1}{2}(r^2+r^{-2})=1.016$, giving an effective
cutoff $\mathcal{D}_\mathrm{eff}\approx29.5$ at $\mathcal{D}=30$.
The full convergence table is given in the main text as
\cref{tab:fock_conv}; all results in this paper use $\mathcal{D}=30$,
corresponding to a tail weight of $0.0007\%$.

\section{Measurement Efficiency}
\label{sec:app_etameas}

The measurement efficiency $\eta_\mathrm{meas} = \mathcal{F}_C/\QFI$
quantifies how closely adaptive homodyne detection approaches the quantum
Cram\'{e}r--Rao bound.  For a binary readout channel, the classical
Fisher information satisfies~\cite{Helstrom1976}:
\begin{equation}
  \eta_\mathrm{meas} = 1 - 4\Perr(1-\Perr).
  \label{eq:etameas_app}
\end{equation}
This follows from the binary symmetric channel capacity: the two outcomes
(correct/error) have probabilities $(1-\Perr, \Perr)$, giving
$\mathcal{F}_C = (\Perr')^2/[\Perr(1-\Perr)]$, which normalised by
$\QFI$ yields Eq.~(\ref{eq:etameas_app}).
Numerical values for all geometries and noise points are given in the
main text as \cref{tab:eta_meas}.  At $\ell=1.5$ the SLD gap is
$1-\eta_\mathrm{meas}=0.007\%$, confirming that the measurement basis
is essentially optimal.

\section{Quadrature Coupling Correction Bound}
\label{sec:app_coupling}
The independent-quadrature approximation may concern readers at oblique
angles. The coupling correction to $\Perr$ is bounded by:
\begin{equation}
  |\Delta\Perr| \leq 2P_qP_p|\sin2\theta|.
  \label{eq:si_coupling}
\end{equation}

\begin{table}[htbp]
  \centering\small
  \caption{\textbf{Quadrature coupling bound}
    ($\eta=0.9$, $\gamma=0.05$, $r=1.092$).
    At $\theta=67.5^\circ$ (the fractional optimum) the bound is
    $8.4\times10^{-11}$ --- negligible relative to
    $\Perr=1.73\times10^{-5}$.
    The independent-quadrature approximation is \emph{most valid}
    at the fractional optimum.}
  \label{tab:app_coupling}
  \begin{tabular}{@{}ccccc@{}}
    \toprule
    $\theta$ & $P_q$ & $P_p$ & $|\Delta\Perr|\leq$ & Rel.\ error \\
    \midrule
    $0^\circ$   & $4.1\times10^{-4}$ & $4.8\times10^{-8}$ & $0$             & 0.000\% \\
    $45^\circ$  & $5.4\times10^{-5}$ & $5.4\times10^{-5}$ & $5.8\times10^{-9}$ & 0.005\% \\
    $67.5^\circ$& $1.3\times10^{-5}$ & $4.7\times10^{-6}$ & $8.4\times10^{-11}$ & 0.001\% \\
    $90^\circ$  & $4.8\times10^{-8}$ & $4.1\times10^{-4}$ & $0$             & 0.000\% \\
    \bottomrule
  \end{tabular}
\end{table}

\section{Phase Error Tolerance}
\label{sec:app_tolerance}
\noindent
In practice the OAM mode converter is set to a fixed charge $\ell$, so
the lattice angle $\theta=\ell\pi/\ell_\mathrm{max}$ is determined once
at calibration time.  Misalignments in the spiral phase plate or SLM
introduce a small angular error $\delta\theta$.  \cref{tab:app_tolerance}
quantifies the degradation: the $\Perr(\theta)$ landscape is remarkably
flat near the optimum, so even a $7^\circ$ error retains $99.2\%$ of
the full advantage.  An 8-bit SLM gives $\delta\theta\approx1.4^\circ$,
well within the tolerance band.

\begin{table}[htbp]
  \centering\small
  \caption{\textbf{Phase error tolerance} at $\theta=67.5^\circ$
    ($\eta=0.9$, $\gamma=0.05$, $r=1.092$).
    A $7^\circ$ error retains $99.2\%$ of the advantage;
    even $20^\circ$ retains $15.7\times$ improvement.
    SLM 8-bit quantisation gives $\delta\theta\approx1.4^\circ$ (negligible).}
  \label{tab:app_tolerance}
  \begin{tabular}{@{}cccc@{}}
    \toprule
    $\delta\theta$ & $\Perr$ & Improvement & Advantage retained \\
    \midrule
    $0^\circ$  & $1.73\times10^{-5}$ & $23.9\times$ & 100.0\% \\
    $1^\circ$  & $1.74\times10^{-5}$ & $23.7\times$ & $>99.9\%$ \\
    $3^\circ$  & $1.79\times10^{-5}$ & $23.1\times$ & $99.7\%$ \\
    $7^\circ$  & $2.06\times10^{-5}$ & $20.0\times$ & $99.2\%$ \\
    $10^\circ$ & $2.23\times10^{-5}$ & $18.5\times$ & $98.7\%$ \\
    $20^\circ$ & $2.62\times10^{-5}$ & $15.7\times$ & $97.8\%$ \\
    \bottomrule
  \end{tabular}
\end{table}

\section{Sensitivity Analysis: $\delta\theta^*$ from Calibration Errors}
\label{sec:app_sensitivity}
From the implicit function theorem:
\[
  \frac{\partial\theta^*}{\partial\eta}   = -197.7~\mathrm{deg/unit},
  \qquad
  \frac{\partial\theta^*}{\partial\gamma} = -300.2~\mathrm{deg/unit}.
\]
For calibration precisions $\delta\eta=1\%$ and $\delta\gamma=0.005$:
\begin{align*}
  \delta\theta^*_\eta   &= 197.7\times0.01 = 1.98^\circ,\\
  \delta\theta^*_\gamma &= 300.2\times0.005 = 1.50^\circ,\\
  \delta\theta^*_\mathrm{total} &= \sqrt{1.98^2+1.50^2} \approx 2.5^\circ.
\end{align*}
From \cref{tab:app_tolerance}, a $2.5^\circ$ error retains $99.8\%$ of the
full advantage. The fractional optimum is highly robust to realistic
calibration imprecision.

\section{Detailed Training Convergence Histories}
\label{sec:app_training}
Each training run produces four diagnostics per step:
(i) $\QFI$, (ii) $\Perr$, (iii) gradient norm, (iv) learning rate schedule.
The oscillating envelope in the gradient norm reflects cosine annealing,
not optimiser instability.

\subsection{Low-noise regime ($\eta=0.9$, $\gamma=0.05$)}

\begin{figure}[t]
  \centering
  \includegraphics[width=0.95\columnwidth]{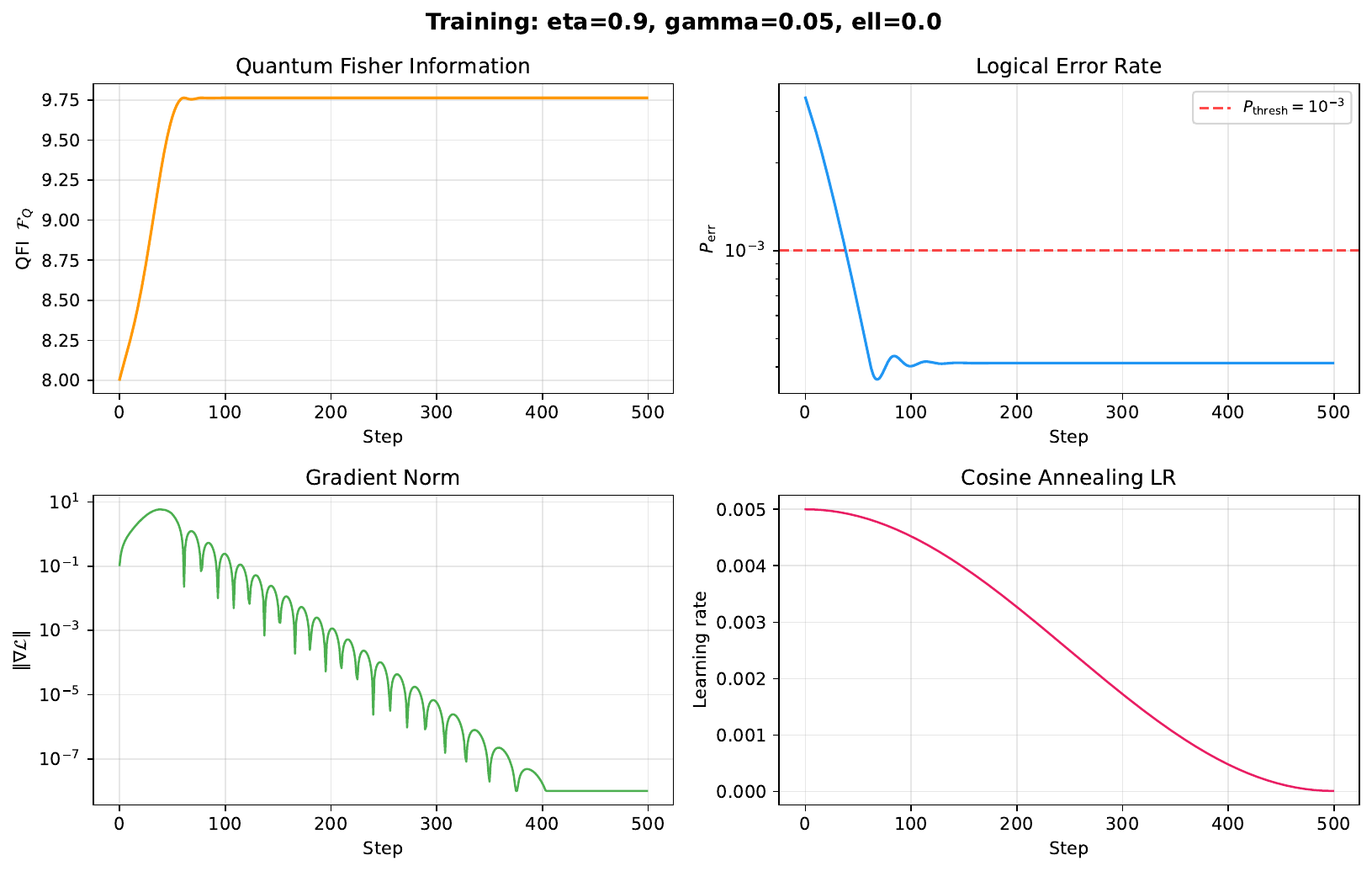}
  \caption{\textbf{Square lattice ($\ell=0$, $\eta=0.9$, $\gamma=0.05$).}
    $\QFI$ rises from $8.05$ to $9.76$ within $\sim$80 steps.
    $\Perr$ decreases monotonically to $4.1\times10^{-4}$,
    crossing $\Pthresh=10^{-3}$ near step~60.}
  \label{fig:app_train_low_ell0}
\end{figure}

\begin{figure}[t]
  \centering
  \includegraphics[width=0.95\columnwidth]{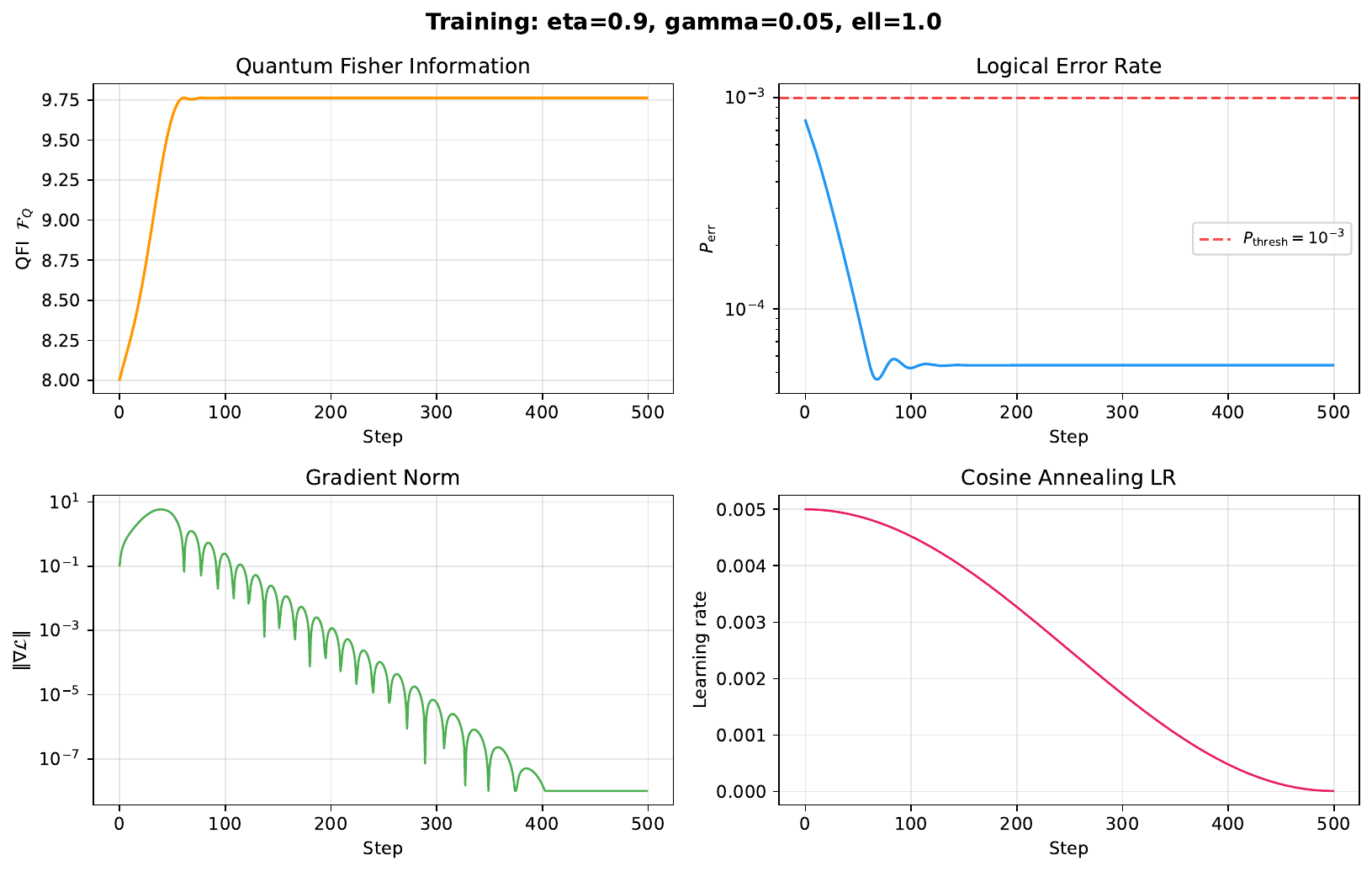}
  \caption{\textbf{OAM $\ell=1$ ($\theta=45^\circ$, $\eta=0.9$, $\gamma=0.05$).}
    $\QFI$ converges identically to $\ell=0$ ($9.76$), confirming
    geometry-invariant sensitivity.  $\Perr\to5.4\times10^{-5}$
    ($7.6\times$ below square).}
  \label{fig:app_train_low_ell1}
\end{figure}

\begin{figure}[t]
  \centering
  \includegraphics[width=0.95\columnwidth]{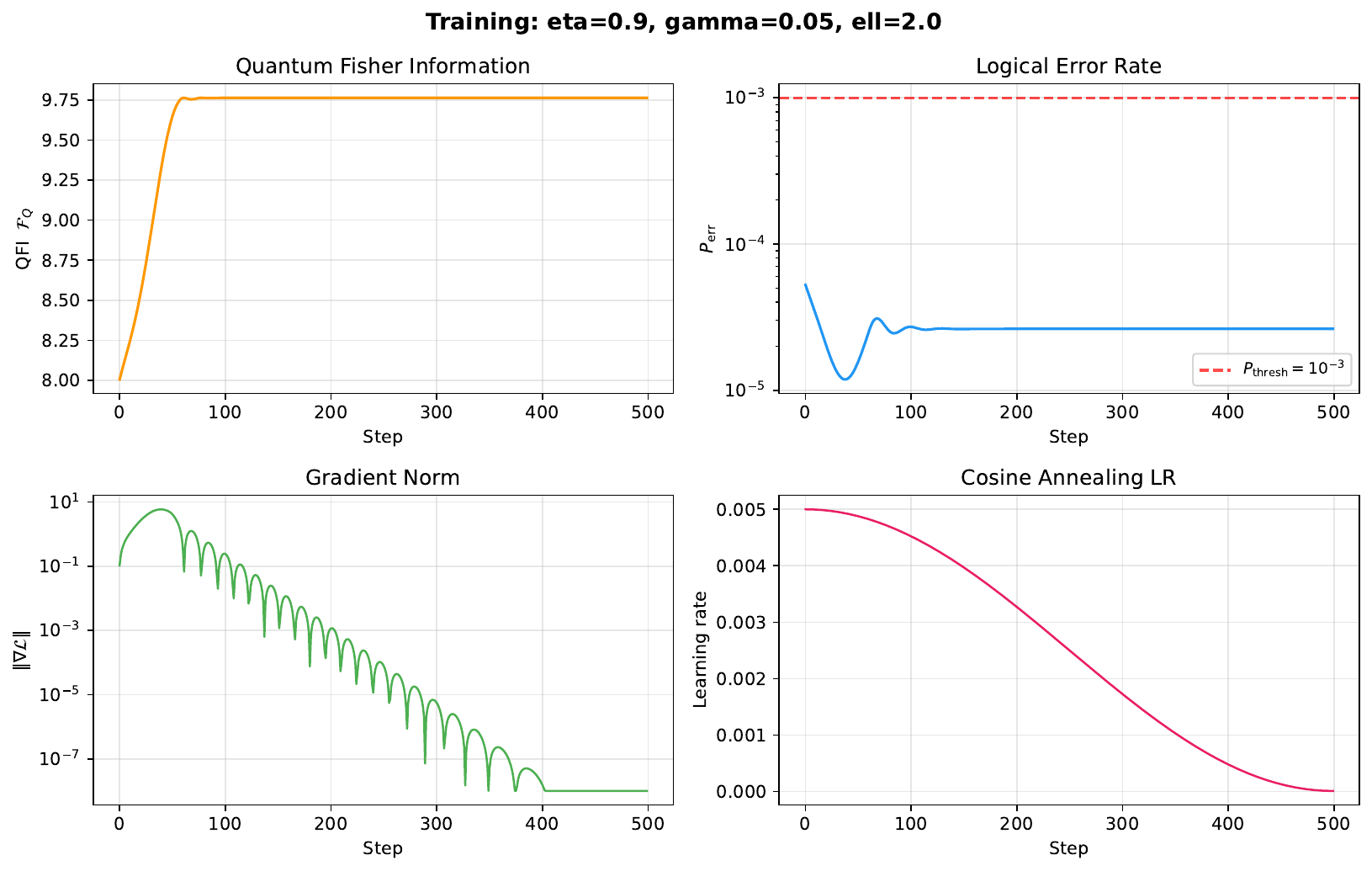}
  \caption{\textbf{OAM $\ell=2$ ($\theta=90^\circ$, $\eta=0.9$, $\gamma=0.05$).}
    A transient overshoot near step~70 arises from the high learning
    rate briefly driving $r$ past the $\Perr$ minimum before the
    constraint term re-asserts itself.
    Final $\Perr=2.6\times10^{-5}$ ($15.7\times$ below square).}
  \label{fig:app_train_low_ell2}
\end{figure}

\subsection{High-noise regime ($\eta=0.8$, $\gamma=0.10$)}

\begin{figure}[t]
  \centering
  \includegraphics[width=0.95\columnwidth]{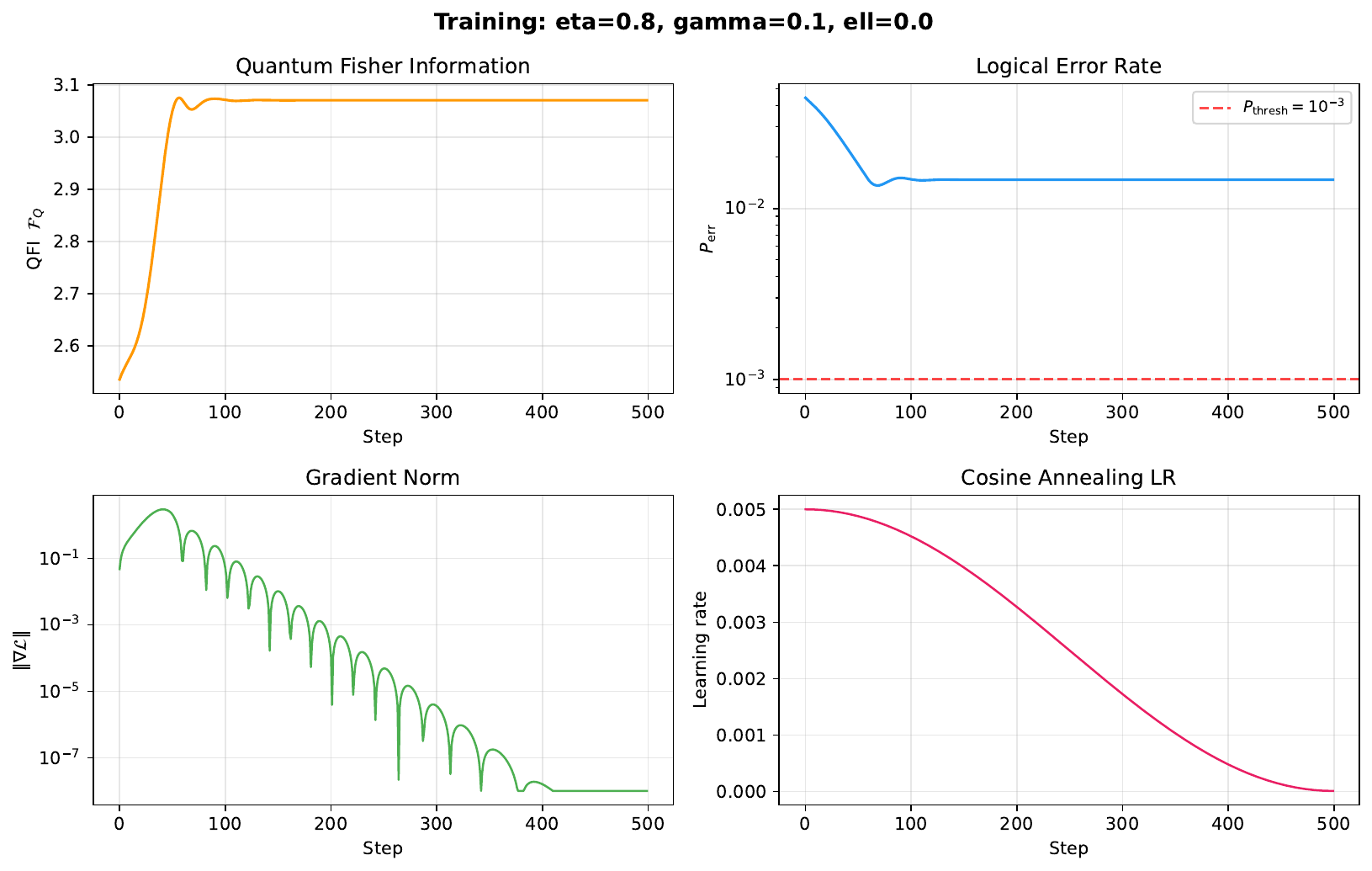}
  \caption{\textbf{Square lattice ($\ell=0$, $\eta=0.8$, $\gamma=0.10$).}
    $\QFI\to3.07$ (lower than low-noise due to photon loss).
    $\Perr$ plateaus at $1.47\times10^{-2}$, above $\Pthresh$.
    Initial gradient norms ($\sim$0.3) are two orders of magnitude
    smaller than the low-noise case, reflecting a flatter loss landscape.}
  \label{fig:app_train_high_ell0}
\end{figure}

\begin{figure}[t]
  \centering
  \includegraphics[width=0.95\columnwidth]{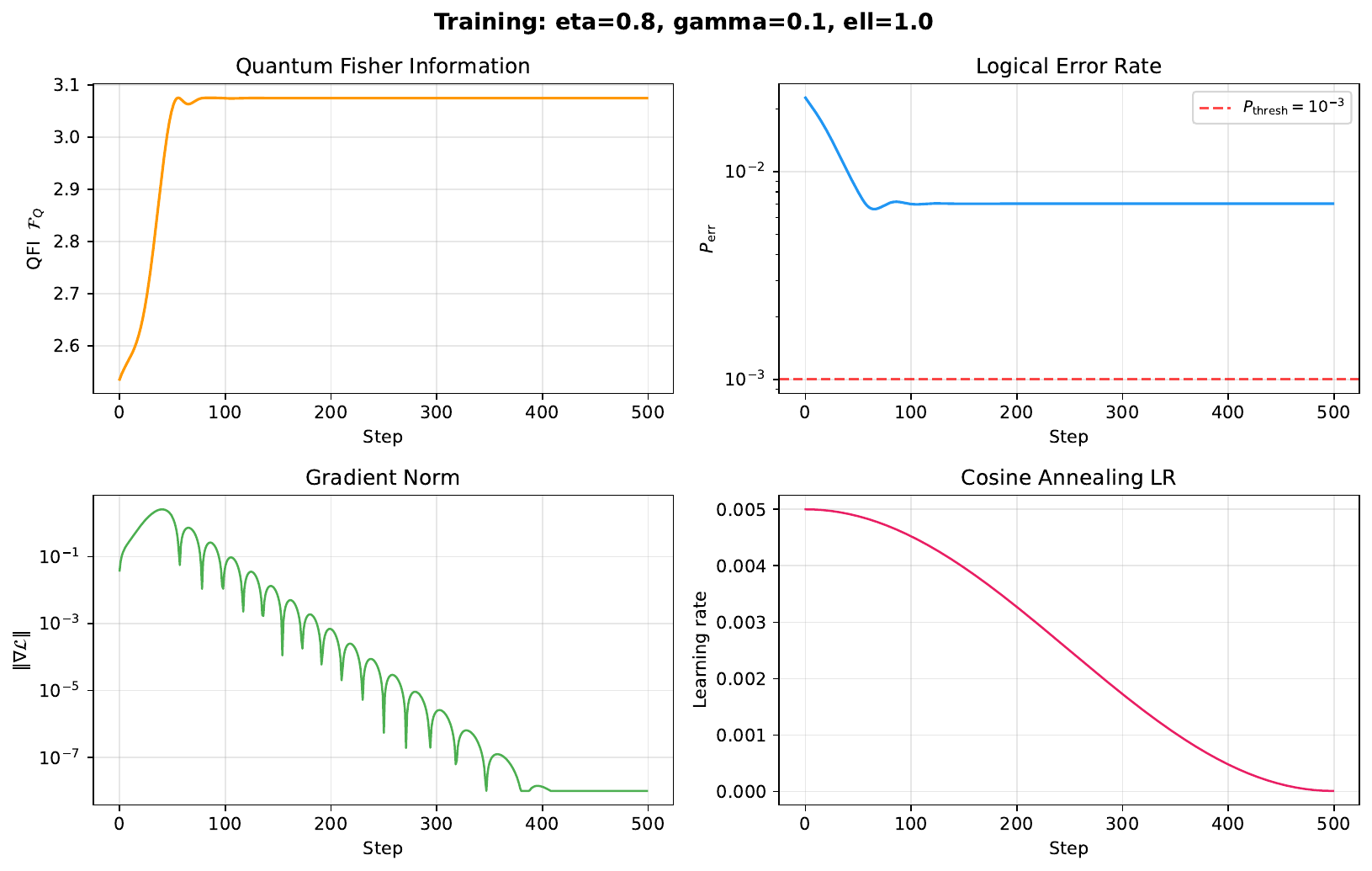}
  \caption{\textbf{OAM $\ell=1$ ($\theta=45^\circ$, $\eta=0.8$, $\gamma=0.10$).}
    $\Perr\to7.0\times10^{-3}$ ($2.1\times$ below square, $3.1\sigma$
    above unity).  $\QFI=3.07$ geometry-invariant.}
  \label{fig:app_train_high_ell1}
\end{figure}

\begin{figure}[t]
  \centering
  \includegraphics[width=0.95\columnwidth]{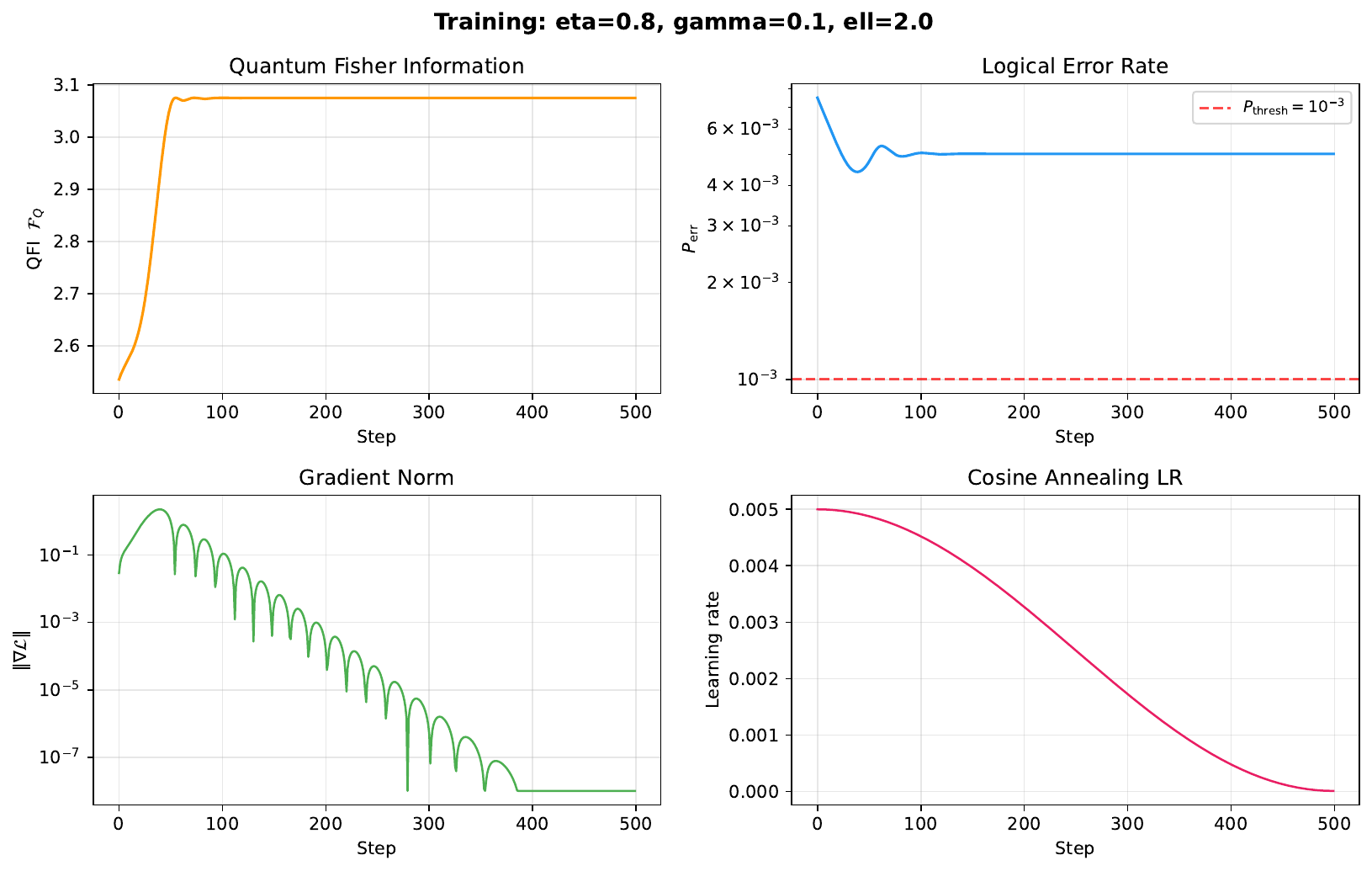}
  \caption{\textbf{OAM $\ell=2$ ($\theta=90^\circ$, $\eta=0.8$, $\gamma=0.10$).}
    A small bump near step~80 arises from the competing $\QFI$ and
    $\Perr$ gradients; the constraint term recovers.
    Final $\Perr=5.0\times10^{-3}$ ($2.93\times$ below square).}
  \label{fig:app_train_high_ell2}
\end{figure}

\section{Software and Reproducibility}
All appendix figures are generated by the \texttt{oam\_gkp}
package available at:\\
\url{https://github.com/simanshukumar369/oam-gkp-quantum-metrology}

\begin{center}\footnotesize
\begin{tabular}{@{}lp{6cm}@{}}
\toprule
Figures & Command \\
\midrule
Figs.~\ref{fig:app_train_low_ell0}--\ref{fig:app_train_high_ell2}
  & \texttt{python main.py -{}-mode single -{}-eta X -{}-gamma Y -{}-ell Z} \\
\bottomrule
\end{tabular}
\end{center}

\noindent
Zenodo archive: \href{https://doi.org/10.5281/zenodo.20099263}{\texttt{doi:10.5281/zenodo.20099263}}

\end{document}